\documentclass[pdftex,pra,aps,onecolumn,twoside,nofootinbib,showkeys]{revtex4}
\usepackage{amssymb}
\usepackage{amsfonts}
\usepackage{enumerate}
\usepackage{float}
\usepackage{color}
\usepackage{indentfirst}
\usepackage{url}
\usepackage[stable]{footmisc}
\usepackage{fancyhdr}
\usepackage{appendix}
\usepackage{dsfont}
\usepackage[all]{xy}
\usepackage[pdftex]{graphicx}
\usepackage{graphicx}
\usepackage{youngtab}
\usepackage[T1]{fontenc}
\usepackage{textcomp}
\usepackage{newcent}
\usepackage[sc,noBBpl]{mathpazo}
\usepackage{amsmath,amsfonts,amssymb,amsthm}
\usepackage[mathscr]{eucal}
\newtheorem{theorem}{Theorem}

\newtheorem{corollary}[theorem]{Corollary}

\newtheorem{definition}[theorem]{Definition}
\newtheorem{example}[theorem]{Example}

\newtheorem{lemma}[theorem]{Lemma}
\newtheorem{notation}[theorem]{Notation}

\newtheorem{proposition}[theorem]{Proposition}
\newtheorem{remark}[theorem]{Remark}

\newcommand\be{\begin{equation}}
\newcommand\ee{\end{equation}}
\newcommand\Span{\operatorname{span}}
\newcommand\Hom{\operatorname{Hom}}
\newcommand\ind{\operatorname{ind}}
\newcommand\id{\operatorname{id}}
\newcommand\sgn{\operatorname{sgn}}
\newcommand\rank{\operatorname{rank}}
\newcommand\tr{\operatorname{tr}}

\begin{document}
\reversemarginpar

\title{Structure and properties of the algebra of partially transposed permutation operators}

\author{Marek Mozrzymas$^1$, Micha{\l} Horodecki$^{2,3}$ and Micha{\l} Studzi{\'n}ski$^{2,3}$}

\affiliation{
$^1$Institute for Theoretical Physics,
University of Wroc{\l}aw, 50-204 Wroc{\l}aw, Poland\\
$^2$Institute for Theoretical Physics and Astrophysics,
University of Gda{\'n}sk, 80-952 Gda{\'n}sk, Poland\\
$^3$National Quantum Information Centre of Gda\'{n}sk, 81-824 Sopot, Poland}

\date{\today}

\begin{abstract}
We consider the structure of algebra of operators, acting in $n-$fold tensor product
space, which are partially transposed on the last term. Using purely algebraical
methods we show that this algebra is semi-simple and then, considering its regular
representation, we derive basic properties of the algebra. In particular, we describe
all irreducible representations of the algebra of partially transposed operators and
derive expressions for matrix elements of the representations. It appears that there
are two types of irreducible representations of the algebra. The first one is strictly
connected with the representations of the group $S(n-1)$ induced by irreducible
representations of the group $S(n-2)$. The second type is structurally connected
with irreducible representations of the group $S(n-1)$. 
\end{abstract}


\maketitle 

\section{ Introduction}
\label{introduction}
In this manuscript we present strong mathematical tool to study PPT property for certain class of states. Namely we consider  case when some operator is invariant under $U^{\otimes (n-k)}\otimes (U^*)^{\otimes k}$ transformations (where  $*$ denotes complex conjugation) and then it can be decomposed in terms of partially transposed permutation operators~\cite{Zhang1}.
Thanks to this we see that to full analysis it is enough to investigate properties algebra  of partially transposed permutation operators on last $k$ subsystems.  We present full solution for the simplest but nontrivial case $k=1$. Namely we present how to construct irreducible representations of the algebra of partially transposed permutation operators for an arbitrary number of subsystems $n$ and an arbitrary dimension $d$ of local Hilbert space  $\mathcal{H}$.

The concept of partial transposition plays an important role in the theory
of quantum information. The most known example is the PPT-criterion (\textsl{Positive Partial Transpose}) or Peres-Horodecki criterion~\cite{Horodecki2,Peres1} which gives us necessary and sufficient condition for separability in bipartite case when dimensions of subsystems are $2 \times 2$ and $2 \times 3$. From these works we know that whenever spectrum of partially transposed bipartite density operator is positive our state is separable. For higher dimensions and multipartite case Peres-Horodecki criterion gives only necessary condition, but still concept of partial transposition is worth to investigate. When we are dealing with states  having some symmetries problem of separability becomes much more easier to analysis. In particular Eggeling and Werner in~\cite{Eggeling1} present result on separability properties for tripartite states which are $U^{\otimes 3}$ invariant using PPT  property and tools from group theory for an arbitrary dimension of subsystem space. 
In~\cite{Tura1, Augusiak1} authors present solution on open problem of existence of four-qubit entangled symmetric states with positive partial transposition and generalize them to systems consisting an arbitrary number of qubits. Namely authors provide criteria for separability of such states formulated in terms of their ranks. PPT property turned out also relevant for a problems in computer science: it is relaxation of some complexity problem, which can be written in terms of separability~\cite{Montanaro1,Brandao1}. 

A more concrete application of the results of this paper is given in~\cite{Studzinski2}. Namely in this paper the authors present
description of universal quantum cloning machines using techniques presented here and in~\cite{Studzinski1}. They connect
quantity which describes quality of the clones (so called fidelity) with matrix representations of above mentioned
algebra. This result gives new insight into mathematical structure of quantum cloning machines and quantum
cloning in general.
One more concrete motivation behind the present study is the following: the future generalization of this paper,
i.e. extension main results to $U^{\otimes (n-k)}\otimes (U^*)^{\otimes k}$ case, for an arbitrary $k$ might allow us to find analytical expressions
for output R{\'e}nyi entropy for two copies of channel coming from some subspaces suggested by Schur-Weyl~\cite{Boerner}
duality for some fixed number of subsystems $n$ and an arbitrary dimension of Hilbert space $d$, which is relevant
for violation of additivity of minimum output entropy.

The algebra of partially transposed permutation operators, which will be denoted $A_{n}^{t_{n}}(d)$, is the algebra of operators representing the elements of the symmetric group $S(n)$ acting permutationally on a basis of tensor
product space $\mathcal{(\mathbb{C}}^{d})^{\otimes n}$, which are transposed on the last term of the tensor
product. From the definition of this algebra it follows that it is
generated, in the natural way, by partially transposed operators representing
permutations of the group $S(n)$, which will be denoted $V_{d}(\sigma
)^{t_{n}}:\sigma \in S(n).$ The important feature of the algebra $%
A_{n}^{t_{n}}(d),$ is the fact that it contains a subalgebra $A_{n-1}(d)$,
generated by operators representing the subgroup $S(n-1)\subset S(n)$ ,
which are not changed by the partial transposition and which act on the
basis of the space $\mathcal{(\mathbb{C}}^{d})^{\otimes n}$ in the standard permutational way (these operators will
be denoted $V_{d}(\sigma ):\sigma \in S(n-1)$). From this it follows that
the algebra $A_{n}^{t_{n}}(d)$ is a sum of two
subspaces 
\[
A_{n}^{t_{n}}(d)=M+A_{n-1}(d),
\]%
where the subspace $\ M$ is an ideal generated by operators $V_{d}(\sigma
)^{t_{n}}$ representing permutations $\sigma \in S(n)$, which permute the
number $n.$ The operators generating the ideal $M$ are non-invertible,
which shows that the partial transposition changes strongly the properties
of the generating operators. This ideal $M$ 
is the most important in our studies (therefore we call it $M$
as main) because, as we will show, it contains all non-trivial irreducible
representations of the algebra $A_{n}^{t_{n}}(d)$, where by
non-trivial representation we mean a representation in which the
non-invertible generators of $A_{n}^{t_{n}}(d)$ are represented by
operators not equal to zero. The natural generators of the algebra $%
A_{n}^{t_{n}}(d)$ may be linearly dependent or linearly independent, which
depends on the relation between $n$ and $d.$ The structure the algebra $%
A_{n}^{t_{n}}(d)$ and in particular, the structure of its irreducible
representations, also depends on the relation between $n$ and $d$. The
properties of the algebra $A_{n}^{t_{n}}(d)$ were considered in our previous
paper~\cite{Studzinski1}, where we have described the irreducible representations of
the algebra $A_{n}^{t_{n}}(d)$ in the case when the natural generators of $%
A_{n}^{t_{n}}(d)$ are linearly independent. We have described also the
irreducible representations of the algebra $A_{n}^{t_{n}}(d)$ in some
particular cases when the generators of $A_{n}^{t_{n}}(d)$ were linearly
dependent. These results were obtained using the representation approach, in
which we considered the action of the operators of the algebra $%
A_{n}^{t_{n}}(d)$ on a basis of the natural representation space $\mathcal{(%
\mathbb{C}
}^{d})^{\otimes n}.$

In this paper we apply a different, purely algebraical methods. First we
derive the formula for the multiplication rule for the natural generators of
the algebra $A_{n}^{t_{n}}(d)$ and then, treating the algebra in an abstract
way, we consider its properties. In particular we show that the algebra $%
A_{n}^{t_{n}}(d)$ is semi-simple. From the semi-simplicity of the algebra $%
A_{n}^{t_{n}}(d),$ it follows that it is a direct sum of matrix ideals,
which are in turns, a direct sums of all left minimal ideals of $A_{n}^{t_{n}}(d),$
or equivalently, of all irreducible representations of $A_{n}^{t_{n}}(d).$
Considering the left regular representation of the algebra $A_{n}^{t_{n}}(d)$
we construct all matrix ideals of $A_{n}^{t_{n}}(d)$ and next we describe
all irreducible representations of this algebra. From our results it follows
that the irreducible representations of the algebra $A_{n}^{t_{n}}(d)$ are
of  two types, which differ structurally. In the first type, included in the ideal $M,$ the irreducible
representations of $A_{n}^{t_{n}}(d)$ are indexed by irreducible
representations of the group $S(n-2)$ and they are strictly connected with
the representations of the group $S(n-1)$ induced by these irreducible
representations of $S(n-2).$ In these representations, when the condition  $d>n-2$ is satisfied, the 
elements $V_{d}(\sigma ):\sigma \in S(n-1)$  of the algebra $%
A_{n}^{t_{n}}(d)$ are represented as in   the representations of
the group $S(n-1)$ induced by irreducible representations of $S(n-2)$ and the dimension of such a representation of the first
type of $A_{n}^{t_{n}}(d)$ is equal to the dimension the induced
representation of $S(n-1).$  The non-invertible generators of the
ideal $M$ are represented in these representations by non-trivial
and rather complicated formulas. When $d\leq n-2$, the
situation is more complicated, because in this case some of the irreducible
representations of the first type may
be defined on some subspace of the representation space of induced
representation of $S(n-1)$. From the semi-simplicity
of the algebra $A_{n}^{t_{n}}(d)$ it follows that there exists in %
$A_{n}^{t_{n}}(d)$ an ideal $S$ such that %
\[
A_{n}^{t_{n}}(d)=M\oplus S.
\]%
The ideal $S$ contains irreducible representations of the
algebra $A_{n}^{t_{n}}(d)$ of different structure then the
representations of the first type included in the ideal $M.$ These
representations of the second type are indexed by some irreducible
representations of the group $S(n-1).$  The generators $V(\sigma
):\sigma \in S(n-1)$  of the algebra $A_{n}^{t_{n}}(d)$ 
are represented naturally by operators of irreducible representations of $%
S(n-1)$, whereas the generators of the ideal $M$  in $%
A_{n}^{t_{n}}(d)$  e.i. partially transposed operators, are
represented trivially by zero operators. So in the
representations of the second type only
the subalgebra $A_{n-1}(d)$ of $A_{n}^{t_{n}}(d)$ is represented
non-trivially  and the ideal $S$  contains
irreducible representations of the algebra $A_{n}^{t_{n}}(d)$  that
may be called semi-trivial.

In the case when the natural generators of the algebra $A_{n}^{t_{n}}(d)$
are linearly independent the results obtained in this paper coincide with the
results obtained by different method in our previous paper~\cite{Studzinski1}.

At the end of this section we note connection between algebra $A_{n}^{t_{n}}(d)$ and Walled Brauer Algebra~\cite{Benkart1,Koike1,Turaev1,Cox1} (subalgebra of Brauer Algebra~\cite{Brauer1,Pan1,Gavarini1}). Namely algebra of partially transposed permutation operators is a representation of Walled Brauer Algebra~\cite{Zhang1}. From the paper~\cite{Brundan1} we know that whenever $d>n-1$ the dimension of Walled Brauer Algebra is equal to $n!$ and dimension of $A_{n}^{t_{n}}(d)$ is also equal to $n!$ . So in this case this two algebras are isomorphic. When condition $d>n-1$ is not fulfilled we have $\operatorname{dim}A_{n}^{t_{n}}(d) <n!$ while dimension of Walled Brauer Algebra is still equal to $n!$ - we do not have isomorphic between this two algebras. 

One important implication of lack of isomorphism is the issue of semisimplicity.
Translating necessary and sufficient condition from~\cite{Cox1} into language of number of systems and local dimensions of the Hilbert space we obtain that Walled Brauer Algebra is semisimple iff $d>n-2$ and also from the same work we know how to label irreducible components. 
For $d=n-1$ both the algebra $A_{n}^{t_{n}}(d)$ and Walled Brauer Algebra are not isomorphic anymore, but they are still semisimple. When condition $d>n-2$ is not satisfy, then  algebra of partially transposed permutation operators $A_{n}^{t_{n}}(d)$ is still semisimple while Walled Brauer Algebra is not.

Summarizing, from previous works~\cite{Benkart1,Koike1,Turaev1,Cox1} we know conditions when Walled Brauer Algebra is simply reducible and we know how to label their irreducible components, but we do not know how to construct matrix elements of irreducible representations in our case (algebra of partially transposed permutation operators). This work gives solution of this problem for the simplest case, namely when partial transposition is taken over last subsystem. We have solved this problem also for the case when algebra $A_{n}^{t_{n}}(d)$ is not isomorphic to Walled Brauer Algebra. 
Namely in our paper it is shown that, the irreducible representations of the algebra $A_{n}^{t_{n}}(d)$ are labelled by representations of the group $S(n-1)$ induced by irreducible representations of $S(n-2)$ in the first kind and by irreducible representations of $S(n-1)$ in the case of the second kind of irreducible representations of the algebra $A_{n}^{t_{n}}(d)$. The matrix forms of these irreducible representations of the algebra $A_{n}^{t_{n}}(d)$ are expressed by matrices of irreducible representations of the groups $S(n-1)$ and $S(n-2)$.

Our paper is organized in the following way. In section~\ref{preliminaria} we remind briefly
the basic concepts and results in theory of groups, complex
finite-dimensional algebras and their representations. In section~\ref{partially_trans} we
introduce the algebra of partially transformed operators $A_{n}^{t_{n}}(d)$
and derive some its properties. The main results of this paper is presented
in section~\ref{irreps1}. In the remaining sections and appendices of the paper contain
the derivation of the main results. In particular in section 5 we give the
construction of the matrix ideals and minimal left ideals (irreducible
representations) of the algebra $A_{n}^{t_{n}}(d).$

\section{Preliminaria}
\label{preliminaria}
In this paper we will have to deal with the group $S(n)$, finite-dimensional
complex algebras, in particular the group algebra $\mathbb{C} \lbrack S(n)]$, and their representations, therefore in this section we
remind briefly basic properties of these structures that will be applied
later on. These results my be found in~\cite{Fulton,Landsman,Curtis,Littlewood}. In the
following we will use terminology of modules and representations, which in
fact, are equivalent, but in case of algebras we will use rather concept of
of modules and in case of group we will use the terminology of
representations. The representations of finite groups and left modules of
their group algebras are strictly connected, namely we have

\begin{theorem}
\label{th1}
There is a one-to-one correspondence between finite-dimensional complex representations of a finite-dimensional algebras and the left modules of such algebras. In particular it holds for any group algebra of a finite group and any matrix representation of the group $G$ defines, in a unique way, a matrix representation of the its group algebra $\mathbb{C}[G]$.
\end{theorem}

If the representation of the group is unitary, then the corresponding
representation of the group algebra, as an algebra of operators, is a $
C^{\ast }-$algebra. The algebra of partially transformed operators, that we
will have to deal with is in fact, as we will see, a semi-simple algebra.
Therefore we remind here the definition and the basic properties of
semisimple algebras.

\begin{definition}~\cite{Littlewood}
\label{df2}
The algebra $A$ is semisimple if it not possesses properly nilpotent~\footnote{$x\in A$ is a nilpotent element of algebra $A$ if $x^n=0$ for some $n \in \mathbb{N}$.}
elements other then zero. An element $a\in A$ is properly nilpotent if $
\forall x\in A$ the elements $xa$ and $ax$ are a nilpotent elements.
Note however, that if $xa$ is nilpotent then $ax$ is also nilpotent.
\end{definition}

We have the following criterium for the semisimplicity.

\begin{theorem}
\label{th3}
A finite-dimensional complex algebra is semi-simple iff it is a direct sum
of ideals, such that each one is isomorphic with a matrix algebra. We will
call such ideals matrix ideals.
\end{theorem}

If the algebra is a $C^{\ast }-$algebra then semi-simplicity follows from
its structure , and we have

\begin{theorem}~\cite{Landsman}
\label{th4}
Every finite-dimensional $C^{\ast }-$algebra is a direct sum of matrix
algebras and consequently is a semisimple algebra.
\end{theorem}

The semi-simplicity determines the structure and properties of irreducible
representations of the algebra in the following way

\begin{theorem}~\cite{Littlewood}
\label{th5}
If an associative algebra $A$ over $\mathbb{C}$ is semi-simple then every irreducible $A$-module (i.e. every irreducible
representation of $A$) is isomorphic to some left minimal (i.e. irreducible)
ideal of $A.$ Moreover any left minimal ideal $I$ of $A$ is of the form%
\[
I=Ae=\{ae:a\in A\},
\]
where $e\in A$ is a primitive idempotent i.e. $e^{2}=e$ and $e$ is not the
sum if two idempotents $a,b\in A$ such that $ab=0.$
\end{theorem}

\begin{example}
\label{ex6}
\bigskip The matrix algebra 
\[
M(n,\mathbb{C})=\Span_{\mathbb{C}}\{e_{ij}:e_{ij}e_{kl}=\delta _{jk}e_{il},\quad i,j=1,..,n\} 
\]
is a semi-simple algebra of dimension $n^{2}$ and the irreducible $M(n,\mathbb{C})$ left-modules (representations) are generated by primitive idempotents $e_{ii},$ 
$i=1,...,n$ and are isomorphic to the space $\mathbb{C}^{n}.$ The algebra $M(n,\mathbb{C})$ endowed with hermitian conjugation is a $C^{\ast }$-algebra.
\end{example}

\begin{proposition}
\label{prop7}
\bigskip The structure of the matrix ideals (e.i the ideals that are
generated by elements $E_{ij}:\quad i,j=1,..,n$ which satisfy matrix
multiplication) which appears in the decomposition of a semisimple algebra
in the statement of Th.~\ref{th3} is the same as the structure of the matrix algebras
described in the above Ex.~\ref{ex6}. They are direct sums of left minimal ideals of
the algebra generated by the primitive idempotents (the "diagonal" elements) 
$E_{ii},$ $i=1,...,n$ of the ideal.
\end{proposition}

The construction of matrix ideals of the algebra representing a group
algebra (such an algebra is always semisimple) is described in the Appendix~\ref{AppC}, where the statement of Th.~\ref{th81}~\ref{IV}) is an example for of the Proposition~\ref{prop7}.

Thus from the above it follows that in order to find all irreducible
representations of a semisimple algebra $A$ we have to look for the matrix
ideals which contain left minimal (i.e. irreducible) ideals of $A$ generated
by the primitive idempotents of the ideals. In our considerations of the
properties of the irreducible representations of the algebra of operators
with one partial transposition we will construct, in Section~\ref{construction1}, such a matrix
ideals, which contains all irreducible representations of the algebra.

In our studies of irreducible representations of the algebra of partially
transposed operators the concept of induced representation of a group will
play an important role. \ In the matrix form the induced representation $%
\Psi $ of a group $G$ induced by a representation $\varphi $ of the subgroup 
$H$ is defined in the following way

\begin{definition}~\cite{Curtis}
\label{def8}
Let $\varphi :H\rightarrow M(n,%
\mathbb{C}
)$ be a matrix representation of a subgroup $H$ of the group $G.$ \ Then the
matrix form of the induced representation $\pi =\ind_{H}^{G}(\varphi )$ of a
group $G$ induced by an irrep. $\varphi $ of the subgroup $H\subset G$ has
the following block matrix form 
\[
\forall g\in G\quad \pi _{ai,bj}(g)=(\widehat{\varphi }%
_{ij}(g_{a}^{-1}gg_{b})), 
\]%
where $g_{a},$ $a=1,...,[G:H]$ are representatives of the left cosets $G/H$
and 
\[
\widehat{\varphi }_{ij}(g_{a}^{-1}gg_{b}) = \left\{ \begin{array}{ll}
\varphi _{ij}(g_{a}^{-1}gg_{b}) & \textrm{if $ \ g_{a}^{-1}gg_{b}\in H$,}\\
0 & \textrm{if $ \ g_{a}^{-1}gg_{b}\notin H$}.
\end{array} \right.
\]
\end{definition}

Any induced representation of a given group $G$ may be extended to the
representation of its group algebra $%
\mathbb{C}
\lbrack G].$

In next section we will use the following

\begin{notation}
\label{not9}
Any permutation $\sigma \in S(n)$ defines, in a natural and unique way, two
natural numbers $a,b\in \{1,2,...,n\}$%
\[
n=\sigma (a),\qquad b=\sigma (n) 
\]

Thus we may characterize any permutation by these two numbers in the
following way%
\[
\sigma \equiv \sigma _{(a,b)}\equiv \sigma_{ab}. 
\]%
Note that in general $a,b$ may be different except the case, when one of
them is equal to $n,$ because in this case we have 
\[
a=n\Leftrightarrow b=n. 
\]%
When $a=n=b,$ then $\sigma (n)=n$ and we will use abbreviation $\sigma
=\sigma _{(n,n)}\equiv \sigma _{n}\in S(n-1) \subset S(n).$
\end{notation}

\begin{example}
\label{ex10}
\bigskip 
\[
\left( 
\begin{array}{ccc}
1 & 2 & 3 \\ 
2 & 3 & 1%
\end{array}%
\right) =\left( 
\begin{array}{ccc}
1 & 2 & 3 \\ 
2 & 3 & 1%
\end{array}%
\right) _{(2,1)},\qquad \left( 
\begin{array}{cccc}
1 & 2 & 3 & 4 \\ 
2 & 1 & 4 & 3%
\end{array}%
\right) =\left( 
\begin{array}{cccc}
1 & 2 & 3 & 4 \\ 
2 & 1 & 4 & 3%
\end{array}%
\right) _{(3,3)}. 
\]
\end{example}

The irreducible representations of the symmetric group $S(m)$, which are
uniquely characterized by partitions $\lambda =(\lambda _{1},\lambda
_{2},..,\lambda _{k})$ or equivalently by Young diagrams $Y(\lambda )$ will
be denoted by greek symbols $\varphi ^{\lambda }$, $\psi ^{\nu }$ etc.

\section{The algebra of partially transposed operators.}
\label{partially_trans}
Let us consider a representation $V$ of the group $S(n)$ in the space 
 $\mathcal{H\equiv (%
\mathbb{C}
}^{d})^{\otimes n}$ defined in the following way

\begin{definition}
\label{def11}
$V:$ $S(n)$ $\rightarrow \Hom(\mathcal{(%
\mathbb{C}
}^{d})^{\otimes n})$ and 
\[
\forall \sigma \in S(n)\qquad V(\sigma ).e_{i_{1}}\otimes e_{i_{2}}\otimes
..\otimes e_{i_{n}}=e_{i_{\sigma ^{-1}(1)}}\otimes e_{i_{\sigma
^{-1}(2)}}\otimes ..\otimes e_{i_{\sigma ^{-1}(n)}}, 
\]%
where $d\in 
\mathbb{N}
$ and $\{e_{i}\}_{i=1}^{d}$ is an orthonormal basis of the space $\mathcal{%
\mathbb{C}
}^{d}.$
\end{definition}

The representation $V:$ $S(n)$ $\rightarrow \Hom(\mathcal{(%
\mathbb{C}
}^{d})^{\otimes n})$ is defined in a given basis $\{e_{i}\}_{i=1}^{d}$ of
the space $\mathcal{%
\mathbb{C}
}^{d}$ (and consequently in a given basis of $H$), so in fact it is a matrix
representation.

\begin{remark}
\label{rem12}
The representation $V:$ $S(n)$ $\rightarrow \Hom(\mathcal{(%
\mathbb{C}
}^{d})^{\otimes n})$ depends explicitley on the dimension $d$, so in fact we
should write $V\equiv V_{d}$ but for simplicity we will omit the index $d$,
unless it will be necessary.
\end{remark}

The representation $V$ of the group $S(n)$ is unitary and we have 
\[
\forall \sigma \in S(n)\qquad V(\sigma )^{\dagger}=V(\sigma ^{-1}),
\]%
where $\dagger$ denotes usual hermitian conjugation with respect to the scalar product in $
\mathcal{(\mathbb{C}}^{d})^{\otimes n}$.

For $d>1$ the representation $V$ is always reducible and we have

\begin{proposition}~\cite{Fulton}
\label{prop13}
The irreducible representation $\varphi ^{\alpha }$ of $S(n),$ indexed by
the partition $\alpha =(\alpha _{1},\alpha _{2},..,\alpha _{k}),$ is
contained in $V$ if $d\geq k\equiv h(\alpha ).$ In particular if $d\geq n$
then all irreducible representations of $S(n)$ are included in the
representation $V$ of $S(n).$ When $d\geq k\equiv h(\alpha )$ then the
multiplicity of the irreducible representation $\varphi ^{\alpha }$ of $S(n)$
is equal to%
\[
\frac{1}{n!}\sum_{\sigma \in S(n)}\chi ^{\alpha }(\sigma ^{-1})d^{l(\sigma
)}, 
\]%
where $\chi ^{\alpha }$ is the character of $\varphi ^{\alpha }$ , $l(\sigma
)$ is the number of cycles in the permutation $\sigma $ and $\chi
^{V}(\sigma )=d^{l(\sigma )}$ is the character of the representation $V:$ $%
S(n)$ $\rightarrow \Hom(\mathcal{(%
\mathbb{C}
}^{d})^{\otimes n}).$
\end{proposition}

The representation $V$ of $S(n)$ extends in a natural way to the
representation of the group algebra $%
\mathbb{C}
\lbrack S(n)]$ and in this way we get the algebra 
\[
A_{n}(d)=\Span_{%
\mathbb{C}
}\{V(\sigma ):\sigma \in S(n)\}\subset \Hom(\mathcal{(%
\mathbb{C}
}^{d})^{\otimes n}) 
\]%
of operators representing the elements of the group algebra $%
\mathbb{C}
\lbrack S(n)]$. Obviously the dimension of the algebra $A_{n}(d)$ depends on
the dimension $d$ and in general we have $n!=\dim 
\mathbb{C}
\lbrack S(n)]\geq $ $\dim A_{n}(d)$ and therefore the algebra $A_{n}(d)$ is
only homomorphic to $%
\mathbb{C}
\lbrack S(n)]$. From Prop.~\ref{prop13} and Th.~\ref{th84} in App.~\ref{AppC} it follows

\begin{proposition}
\label{prop14}
If $d\geq n$ then the operators $V(\sigma ):\sigma \in S(n)$ are linearly
independent and $\dim A_{n}(d)=n!$, if $d<n$ then the operators are linearly
dependent and the dimension of the algebra $A_{n}(d)$ is smaller then $n!.$
\end{proposition}

The algebra $A_{n}(d)$ contain a natural subalgebra 
\[
A_{n-1}(d)=\Span_{%
\mathbb{C}
}\{V(\sigma _{n}):\sigma _{n}\in S(n-1)\}. 
\]%
Extending anti-linearly the hermitian conjugation $\dagger$ of the representation $%
V[S(n)]$, we endove the algebra $A_{n}(d)$ with a $C^{\ast }-$algebra
structure getting a $C^{\ast }-$algebra. The semi-simplicity of the algebra $%
A_{n}(d)$ follows from the complete reducibility of the representation $V:$ $%
S(n)$ $\rightarrow \Hom(\mathcal{(%
\mathbb{C}
}^{d})^{\otimes n})$ or from the $C^{\ast }-$algebra structure of $A_{n}(d)$.

Our main task in this paper is to study the properties of the following
algebra of partially transposed operators

\begin{definition}
\label{def15}
For $A_{n}(d)=\Span_{%
\mathbb{C}
}\{V(\sigma ):\sigma \in S(n)\}$ we define a new complex algebra 
\[
A_{n}^{t_{n}}(d)=\Span_{%
\mathbb{C}
}\{V(\sigma )^{t_{n}}:\sigma \in S(n)\}\subset \Hom(\mathcal{(%
\mathbb{C}
}^{d})^{\otimes n}), 
\]%
where the symbol $t_{n}$ describes the partial transpose in the last place
in the space $\Hom(\mathcal{(%
\mathbb{C}
}^{d})^{\otimes n}).$ The elements $V(\sigma )^{t_{n}}:\sigma \in S(n)$ will
be called natural generators of the algebra $A_{n}^{t_{n}}(d).$
\end{definition}

Directly from this definition it follows

\begin{proposition}
\label{prop16}
\bigskip a) if $\sigma =\sigma _{n}\in S(n-1)\subset S(n)$ then $V(\sigma
_{n})^{t_{n}}=V(\sigma _{n})$ i.e. $V(S(n-1))\subset A_{n}^{t_{n}}(d),$
where $S(n-1)\subset S(n)$ means the natural embedding: $\sigma =\sigma
_{n}\in S(n-1)\subset S(n)$ then $\sigma (n)=n,$

b) $\dim _{%
\mathbb{C}
}A_{n}(d)=\dim _{%
\mathbb{C}
}A_{n}^{t_{n}}(d)$
\end{proposition}

The second statement in this Proposition follows from the invertibility of
the partial transpose and from the Propositions~\ref{prop14},~\ref{prop16} we get

\begin{corollary}
\label{col17}
\bigskip When $d<n$ then the generating elements $V(\sigma )^{t_{n}}$ $%
:\sigma \in S(n)$ of the algebra $A_{n}^{t_{n}}(d)$ are linearly dependent.
\end{corollary}

\begin{remark}
\label{rem18}
Because the partial transpose does not change the elements $V(\sigma _{n})$ $%
\in $ $V(S(n-1))\subset A_{n}^{t_{n}}(d)$ therefore, in the following, we
will write simply $V(\sigma _{n})$ instead $V(\sigma _{n})^{t_{n}}$ when $%
\sigma _{n}$ $\in $ $S(n-1)$ $.$
\end{remark}

The essential for our studies of the properties of the algebra $%
A_{n}^{t_{n}}(d)$ is to describe the composition law of this algebra$.$ A
rather laborious direct calculation gives the following result

\begin{theorem}
\label{th19}
The elements $V(\sigma )^{t_{n}}:\sigma \in S(n)$ which span the algebra $%
A_{n}^{t_{n}}(d)$ have the following composition rule%
\[
V(\sigma _{n})^{t_{n}}V(\rho _{n})^{t_{n}}=V(\sigma _{n})V(\rho
_{n})=V(\sigma _{n}\rho _{n}), 
\]%
\[
V(\sigma _{n})V(\rho _{(a,b)})^{t_{n}}=V(\sigma _{n}\rho
_{(a,b)})^{t_{n}},\qquad V(\sigma _{(a,b)})^{t_{n}}V(\rho _{n})=V(\sigma
_{(a,b)}\rho _{n})^{t_{n}}, 
\]%
\[
V(\sigma _{(a,b)})^{t_{n}}V(\rho _{(p,q)})^{t_{n}}=d^{\delta _{aq}}V[(\sigma
(q)n)\sigma _{(a,b)}\rho _{(p,q)}(pn)]^{t_{n}}, 
\]%
where $a,b\neq n$ and $c,d\neq n$ and the multiplication rule depends
explicitly on $d.$
\end{theorem}

According to these formulas the composition of operators $V(\sigma
_{(a,b)})^{t_{n}}$ is expressed by composition of standard permutations.

\begin{remark}
\label{rem20}
\bigskip Note that, contrary to the case of the standard algebra $A_{n}(d)$,
the composition rule in the algebra $A_{n}^{t_{n}}(d)$ depends explicitly
on the dimension $d.$ Therefore for different values of $d$ we have to deal
with different algebras.
\end{remark}

\begin{remark}
\label{rem21}
\bigskip The explicit formulas for the multiplication in the algebra $%
A_{n}^{t_{n}}(d)$ allows to consider this algebra, in purely abstract way,
as an algebra generated by the elements $V(\sigma )^{t_{n}}$ $:\sigma \in
S(n)$, which satisfies the multiplication rules stated in the Th.~\ref{th19}. In the
following we will treat the algebra in this way.
\end{remark}

Many particular cases follow almost immediately from the above composition
law.

\begin{example}
\label{ex22}
\[
V(kn)^{t_{n}}V(jn)^{t_{n}}=V[(jn)(kn)]^{t_{n}},\quad k\neq j. 
\]%
\[
V(kn)^{t_{n}}V(kn)^{t_{n}}=dV[(kn)]^{t_{n}}, 
\]%
\[
V(ijn)^{t_{n}}V(ijn)^{t_{n}}=V(ijn)^{t_{n}} 
\]%
so in particular $V(kn)^{t_{n}}$ and $V(ijn)^{t_{n}}$ are (essential)
projectors.
\end{example}
\begin{example}
In the table below we present composition properties for algebra $A_3^{t_3}(d)$.
\begin{table}[h!]
\label{table1}
\begin{tabular}{l|llllll}
$\circ $ & $\text{\noindent\(\mathds{1}\)}$ & $(132)^{t}$ & $(123)^{t}$ & $(12)^{t}$ & $(13)^{t}$
& $(23)^{t}$ \\ 
\hline
$\text{\noindent\(\mathds{1}\)}$ & $\text{\noindent\(\mathds{1}\)}$ & $(132)^{t}$ & $(123)^{t}$ & $(12)^{t}$ & $%
(13)^{t}$ & $(23)^{t}$ \\ 
$(132)^{t}$ & $(132)^{t}$ & $(132)^{t}$ & $d(23)^{t}$ & $(23)^{t}$ & $%
d(132)^{t}$ & $(23)^{t}$ \\ 
$(123)^{t}$ & $(123)^{t}$ & $d(13)^{t}$ & $(123)^{t}$ & $(13)^{t}$ & $%
(13)^{t}$ & $d(123)^{t}$ \\ 
$(12)^{t}$ & $(12)^{t}$ & $(13)^{t}$ & $(23)^{t}$ & $\text{\noindent\(\mathds{1}\)}$ & $%
(132)^{t} $ & $(123)^{t}$ \\ 
$(13)^{t}$ & $(13)^{t}$ & $(13)^{t}$ & $d(123)^{t}$ & $(123)^{t}$ & $%
d(13)^{t}$ & $(123)^{t}$ \\ 
$(23)^{t}$ & $(32)^{t}$ & $d(132)^{t}$ & $(23)^{t}$ & $(132)^{t}$ & $%
(132)^{t}$ & $d(23)^{t}$%
\end{tabular}
\caption{For simplicity here we have $t=t_{3}.$ From this table it follows that $(132)^{t},$ $(123)^{t},$ $%
(13)^{t},(23)^{t}$ \ of the algebra $A_{3}^{t_{3}}(d)$ are idempotents or
essential idempotents. The operator $(12)$, as well $\ \text{\noindent\(\mathds{1}\)},$ remains
unchanged under the transposition on the third position $t_{3}.$}
\end{table}
\end{example}

\bigskip The hermitian conjugation in the space $\Hom(\mathcal{(%
\mathbb{C}
}^{d})^{\otimes n})$ commutes with the partial transpose i.e. we have 
\[
\forall \sigma \in S(n)\quad (V(\sigma )^{t_{n}})^{\dagger}=((V(\sigma
)^{\dagger})^{t_{n}}=V(\sigma ^{-1})^{t_{n}} 
\]%
so the algebra $A_{n}^{t_{n}}(d)$ is invariant under the hermitian
conjugation $\dagger$ and moreover we have

Moreover directly form the definition of the operators $V(\sigma )^{t_{n}}$ we have:
\begin{proposition}
\label{prop23}
The algebra $A_{n}^{t_{n}}(d)$ with the hermitian conjugation is a $C^{\ast
}-$algebra i.e.%
\[
\forall a,b\in A_{n}^{t_{n}}(d)\quad (ab)^{\dagger}=b^{\dagger}a^{\dagger}. 
\]
\end{proposition}

From the Theorem~\ref{th4} we get

\begin{corollary}
\label{col24}
\bigskip The algebra $A_{n}^{t_{n}}(d)$ is a semi-simple algebra for any
value of $d$ and $n.$
\end{corollary}

\bigskip We have also one more consequence of the Theorem~\ref{th19}

\begin{corollary}
\label{col25}
For $a,b\neq n$ the elements $V(\sigma _{(a,b)})^{t_{n}}$ are not invertible
and the set $M=\Span_{%
\mathbb{C}
}\{V(\sigma _{(a,b)})^{t_{n}}:a,b\neq n\}$ is an ideal of the algebra $%
A_{n}^{t_{n}}(d).$ We have the following decomposition of the algebra $%
A_{n}^{t_{n}}(d)$ 
\[
A_{n}^{t_{n}}(d)=M+ \Span_{%
\mathbb{C}
}\{V(\sigma _{n})^{t_{n}}=V(\sigma _{n}):\sigma _{n}\in S(n-1)\}. 
\]
\end{corollary}

Note that the second component in the simple sum,  which in
general, is not simple, is not an ideal of the
algebra $A_{n}^{t_{n}}(d).$ The ideal $M,$ being semi-simple, has a unit
element $e,$ which in fact an idempotent of the algebra $A_{n}^{t_{n}}(d)$
and $M=eA_{n}^{t_{n}}(d)e.$ Using the basic properties of the idempotents $e,$
$\text{\noindent\(\mathds{1}\)}-e,$ one can construct another ideal $S=(\text{\noindent\(\mathds{1}\)}%
-e)A_{n}^{t_{n}}(d)(\text{\noindent\(\mathds{1}\)}-e)\subset A_{n}^{t_{n}}(d)$ which is, as we
will see in next section, closely related to the group algebra $%
\mathbb{C}
\lbrack S(n-1)]$ such that

\begin{proposition}
\label{prop26}
The algebra $A_{n}^{t_{n}}(d)$ has the following decomposition into a direct
sum of two  ideals%
\[
A_{n}^{t_{n}}(d)=eA_{n}^{t_{n}}(d)e\oplus (\text{\noindent\(\mathds{1}\)}-e)A_{n}^{t_{n}}(d)(%
\text{\noindent\(\mathds{1}\)}-e)\equiv M\oplus S, 
\]%
where 
\[
MS=0. 
\]
\end{proposition}

\bigskip This Proposition is the first step in the decomposition of the
algebra $A_{n}^{t_{n}}(d)$ into a direct sum of matrix algebras. In the
Section 5 we will show that the ideals $M$ and $S$ are in fact, the direct
sums of matrix ideals of the algebra $A_{n}^{t_{n}}(d)$, such that each one
is a direct sum of minimal left ideals (irreducible representations) of the
algebra $A_{n}^{t_{n}}(d).$ The construction however, is rather complicated
and laborious and it is the content of the remaining part of this paper,
therefore in the following section we formulate the theorems describing the
structure and irreducible representations of the algebra $A_{n}^{t_{n}}(d),$
which is in fact the main result of the paper.

\section{Main result: the structure and irreducible representations of \ the
algebra $A_{n}^{t_{n}}(d).$}
\label{irreps1}

It appears that the structure of irreducible representations of the algebra $%
A_{n}^{t_{n}}(d)$ is closely related to the structure of the representation $%
\ind_{S(n-2)}^{S(n-1)}(\varphi ^{\alpha })$ of the group $S(n-1)$ induced by
irreducible representations $\varphi ^{\alpha }$ of the group $S(n-2)$ and
the properties of irreducible representations of $A_{n}^{t_{n}}(d)$ depends
strongly on the relation between $d$ and $n$. Before presenting the main
results of this paper we have to describe briefly some object appearing in
the structure of the algebra $A_{n}^{t_{n}}(d),$ in particular the properties
of the induced representation $\ind_{S(n-2)}^{S(n-1)}(\varphi ^{\alpha }).$
The matrix form of such a representations given in Def.~\ref{def8}. The irreducible
representations of the group $S(n-2)$ are characterized by the partitions $%
\alpha =(\alpha _{1},...,\alpha _{k})$ of $n-2,$ which describe also the
corresponding Young diagram $Y(\alpha ).$ \ The representation $%
\ind_{S(n-2)}^{S(n-1)}(\varphi ^{\alpha })$ is completely and simply
reducible i.e. we have~\cite{Fulton}.

\begin{proposition}
\label{prop27}
\[
\ind_{S(n-2)}^{S(n-1)}(\varphi ^{\alpha })=\bigoplus _{\nu }\psi ^{\nu }, 
\]%
where the sum is over all partitions $\nu =(\nu _{1},...,\nu _{k})$ of $n-1,$
such that their Young diagrams $Y(\nu )$ are obtained from $Y(\alpha )$ by
adding, in a proper way, one box.
\end{proposition}

From this Proposition it follows that the induced representation $%
\ind_{S(n-2)}^{S(n-1)}(\varphi ^{\alpha })$ may be described in two bases.
The first one, is the basis of the matrix form of the induced representation
given in Def.~\ref{def8} and it is of the form%
\[
\{e_{i}^{a}(\alpha ):a=1,...,n-1,\quad i=1,...,\dim \varphi ^{\alpha }\},
\]%
where the index $a=1,...,n-1$ describes the the costes $S(n-1)/S(n-2)$ and
the the index $i=1,...,\dim \varphi ^{\alpha }$ is the index of a matrix
form of $\varphi ^{\alpha }.$ The second one is a basis of the reduced form
of $\ind_{S(n-2)}^{S(n-1)}(\varphi ^{\alpha })$, given Prop.~\ref{prop27} which is of the
form%
\[
\{f_{j_{\nu }}^{\nu }:\psi ^{\nu }\in \ind_{S(n-2)}^{S(n-1)}(\varphi ^{\alpha
}),\quad j_{\nu }=1,...,\dim \psi ^{\nu }\}. 
\]

\bigskip The next important objects are the following matrices

\begin{definition}
\label{def28}
For any irreducible representation $\varphi ^{\alpha }$ of the group $S(n-2)$
we define the block matrix 
\[
Q_{n-1}^{d}(\alpha )\equiv Q(\alpha )=(d^{\delta _{ab}}\varphi _{ij}^{\alpha
}[(a \ n-1)(ab)(b \ n-1)])=(Q_{ij}^{ab}(\alpha )), 
\]%
where $a,b=1,...,n-1,\quad i,j=1,...,\dim \varphi ^{\alpha }$ and the blocs
of the matrix $Q(\alpha )$ are labelled by indices $(a,b)$ whereas the
elements of the blocks are labelled by the indices of the irreducible
representation $\varphi ^{\alpha }=(\varphi _{ij}^{\alpha })$ of the group $%
S(n-2)$ and $Q(\alpha )\in M((n-1)w^{\alpha },\mathbb{C}).$ Note that if $a,b \neq n-1 \Rightarrow (a \ n-1)(ab)(b \ n-1)=(ab)$, but if $a=n-1$ and $b \neq n-1$ then $(a \ n-1)(ab)(b \ n-1)=id$, where $id$ denotes identity permutation.
\end{definition}

The matrices $Q(\alpha )$ are hermitian and their structure and properties
are described in the Appendix~\ref{AppA}, where it has been shown, that the
eigenvalues $\lambda _{\nu }$ of the matrix $Q(\alpha )$ are labelled by
the irreducible representations $\psi ^{\nu }\in
\ind_{S(n-2)}^{S(n-1)}(\varphi ^{\alpha })$ and the multiplicity of $\lambda
_{\nu }$ is equal to $\dim \psi ^{\nu }$. The essential for properties of
the irreducible representations of the algebra $A_{n}^{t_{n}}(d)$ is the
fact, that at most one (up to the multiplicity) eigenvalue $\lambda _{\nu }$
of the matrix $Q(\alpha )$ may be equal to zero (Cor.~\ref{col67} in App.~\ref{AppA}).

The structure of the algebra $A_{n}^{t_{n}}(d)$ is the following

\begin{theorem}
\label{th29}
The algebra $A_{n}^{t_{n}}(d)$ is a direct sum of two ideals
\[
A=M\oplus S 
\]
and the ideals $M$ and $S$ has different structures.
\begin{enumerate}[a)]
\item The ideal $M$ is of the form 
\[
M=\bigoplus _{\alpha }U(\alpha ), 
\]
where $U(\alpha )$ are ideals of the algebra $A_{n}^{t_{n}}(d)$
characterized by the irreducible representations $\varphi ^{\alpha }$ of the
group $S(n-2)$, such that $\varphi ^{\alpha }\in V_{d}[S(n-2)]$ and 
\[
U(\alpha )=\Span_{\mathbb{C}}\{u_{ij}^{ab}(\alpha ):a,b=1,...,n-1,\quad i,j=1,...,w^{\alpha }\} 
\]
with
\[
u_{ij}^{ab}(\alpha )u_{kl}^{pq}(\beta )=\delta _{\alpha \beta
}Q_{ik}^{bp}(\alpha )u_{il}^{aq}(\alpha ). 
\]
The ideals $U(\alpha )$ are matrix ideals such that
\[
U(\alpha )\simeq M(\rank Q(\alpha ),\mathbb{C}), 
\]
in particular when $\det Q(\alpha )\neq 0$ we have 
\[
U(\alpha )\simeq M((n-1)\dim \varphi ^{\alpha },\mathbb{C}). 
\]
\item The ideal $S$ has the following structure:
\newline
I) if %
$d\geq n$ 
\[
S\simeq \bigoplus _{\nu }M(\dim \psi ^{\nu },%
\mathbb{C}
) 
\]%
where $\nu $ runs over all irreducible representations of
the group $S(n-1)$\newline
II) If $d<n$ then
\[
S\simeq \bigoplus _{\nu }M(\dim \psi ^{\nu },%
\mathbb{C}
) 
\]%
where now $\nu $ runs over all irreducible representations
of the group $S(n-1)$\, such that $d>h(\nu )$ ($h(\nu )$%
 is defined in Prop.13).
\end{enumerate}
\end{theorem}

\bigskip The matrix ideals contained in the ideals $M$ and $S$ contains all
minimal left ideals i.e. all irreducible representations of the algebra $%
A_{n}^{t_{n}}(d)$. The next theorems describes all these representations.

The structure of the irreducible representations of the algebra $%
A_{n}^{t_{n}}(d)$, included in the ideal $M$, is completely determined by
irreducible representations $\varphi ^{\alpha }$ of the group $S(n-2)$,
therefore we will denote them $\Phi _{A}^{\alpha }.$

\begin{theorem}
\label{th30}
The irreducible representations $\Phi _{A}^{\alpha }$ of the algebra $%
A_{n}^{t_{n}}(d)$ contained in the ideal $U(\alpha )\subset M$ (Th.30) are
indexed by the irreducible representations $\varphi ^{\alpha }$ of the group 
$S(n-2)$, such that $\varphi ^{\alpha }\in V_{d}[S(n-2)]$ and if $%
\{f_{j_{\nu }}^{\nu }:\psi ^{\nu }\in \ind_{S(n-2)}^{S(n-1)}(\varphi ^{\alpha
}),\quad j_{\nu }=1,...,\dim \psi ^{\nu }\}$ is the reduced basis of the
induced representation $\ind_{S(n-2)}^{S(n-1)}(\varphi ^{\alpha })$, then the
vectors $\{f_{j_{\nu }}^{\nu }:\lambda _{\nu }\neq 0\}$ form the basis of
the irreducible representation of the algebra $A_{n}^{t_{n}}(d)$ and the
natural generators of $A_{n}^{t_{n}}(d)$ act on it in the following way%
\be
V(an)^{t_{n}}{}f_{j\nu }^{\nu }(\alpha )=\sum_{\rho ,j_{\rho }}(\sum_{k}%
\sqrt{\lambda _{\rho }}z^{+}(\alpha )_{j_{\rho }k}^{\rho a}z(\alpha
)_{kj_{\nu }}^{a\nu }\sqrt{\lambda _{\nu }}f_{j_{\rho }}^{\rho }(\alpha ) 
\ee
where the summation is over $\rho $ such that $\lambda _{\rho }\neq 0.$ 
Due to the condition $\varphi ^{\alpha }\in V_{d}[S(n-2)]$
the eigen-values $\lambda _{\nu }$ of $\ Q(\alpha )$ are
non-negative. The unitary matrix $Z(\alpha )=(z(\alpha )_{kj_{\nu }}^{a\nu
})$ has the form 
\be
 z(\alpha )_{kj_{\nu }}^{a\nu }=\frac{\dim \psi ^{\nu }}{\sqrt{%
N_{j_{\nu }}^{\nu }}(n-1)!}\sum_{\sigma \in S(n-1)}\psi _{j_{\nu }j_{\nu
}}^{\nu }(\sigma ^{-1})\delta _{a \ \sigma (q)}\varphi _{kr}^{\alpha
}[(a \ n-1)\sigma (q \  n-1)],
\ee
with 
\be
 N_{j_{\nu }}^{\nu }=\frac{\dim \psi ^{\nu }}{(n-1)!}%
\sum_{\sigma \in S(n-1)}\psi _{j_{\nu }j_{\nu }}^{\nu }(\sigma ^{-1})\delta
_{q \ \sigma (q)}\varphi _{rr}^{\alpha }[(q \ n-1)\sigma (q \ n-1)],
\ee
where the indices $q=1,..,n-1,$ $r=1,..,\dim \varphi
^{\alpha }$ are fixed and such that $N_{j_{\nu }}^{\nu }>0$ (see Th.59, 85 and Cor. 86). For $\sigma _{n}\in S(n-1)$ we have 
\[
V(\sigma _{n})f_{j_{\nu }}^{\nu }(\alpha )=\sum_{\rho ,j_{\rho }}\psi
_{i_{\nu }j_{\nu }}^{\nu }(\sigma _{n})f_{i_{\nu }}^{\nu }(\alpha ). 
\]%
In particular when $\det Q(\alpha )\neq 0,$ (e.i. when all $\lambda _{\nu
}\neq 0$) then the representation $\Phi _{A}^{\alpha }$ is the induced
representation $\ind_{S(n-2)}^{S(n-1)}(\varphi ^{\alpha })$ (in the reduced
form) for the sublagebra $V_{d}[S(n-1)]\subset $ $A_{n}^{t_{n}}(d).$ In this
case the dimension of the irreducible representation is equal to 
\[
\dim \Phi _{A}^{\alpha }=(n-1)\dim \varphi ^{\alpha }=\dim
(\ind_{S(n-2)}^{S(n-1)}(\varphi ^{\alpha })). 
\]%
When $\det Q(\alpha )=0,$ (e.i. when one, up to the multilicity, eigen-value 
$\lambda _{\theta }$ of $Q(\alpha )$ is equal to $0$)$,$ then the
irreducible representation of $A_{n}^{t_{n}}(d)$ is defined on a subspace $%
\{y_{j_{\nu }}^{\nu }:\lambda _{\nu }\neq \lambda _{\theta }\}$ of the
representation space $\ind_{S(n-2)}^{S(n-1)}(\varphi ^{\alpha })$ and the
representation has dimension is equal to 
\[
\dim \Phi _{A}^{\alpha }are=\dim (_{S(n-2)}^{S(n-1)}(\varphi ^{\alpha
}))-\dim \psi ^{\theta }=\rank Q(\alpha ). 
\]%
This case takes the place when 
\[
d=i-\alpha _{i}-1 
\]%
for some $\alpha _{i}$ in the partition $\alpha =(\alpha _{1},..,\alpha
_{i},..,\alpha _{k})$ characterizing the irreducible representation $\varphi
^{\alpha }$, under condition that $\nu =(\alpha _{1},..,\alpha
_{i}+1,..,\alpha _{k})$ characterizes the representation $\psi ^{\nu }$ of $%
S(n-1).$

The ideal $U(\alpha )$ is a direct sum of $\dim \Phi _{A}^{\alpha }$ of
irreducible representations $\Phi _{A}^{\alpha }.$
\end{theorem}

The formula for the eigenvalues $\lambda _{\nu }$ of matrices $Q(\alpha )$
are derived in the Appendix~\ref{AppA} (Th.~\ref{th58}).

\begin{remark}
\label{rem31}
\bigskip Note that even if $\dim \varphi ^{\alpha }=1$, we have $\dim \Phi
^{\alpha }=n-1.$
\end{remark}

The matrix forms of these representations are the
following

\begin{proposition}
\label{prop32}
\bigskip In the reduced matrix basis $\{f_{j\nu }^{\nu }:\nu \neq \theta \}$
of the ideal $U(\alpha )$ the natural generators $V(\sigma _{ab})^{t_n }$
and $V(\sigma _{n})$ of $A_{n}^{t_{n}}(d)$ are represented by the following
matrices
\[
M_{f}^{\alpha }(V^{t_{n}}(an))_{j_{\rho }j_{\nu }}^{\rho \nu
}=\sum_{k=1,..,\dim \varphi ^{\alpha }}\sqrt{\lambda _{\rho }}z^{\dagger}(\alpha
)_{j_{\rho }k}^{\rho a}z(\alpha )_{kj_{\nu }}^{a\nu }\sqrt{\lambda _{\nu }}
:\rho ,\nu \neq \theta , 
\]
\[
M_{f}^{\alpha }(V(\sigma _{n}))_{j_{\nu ^{\prime }}j_{\nu }}^{\nu ^{\prime
}\nu }=\delta ^{\nu ^{\prime }\nu }\psi _{j_{\nu ^{\prime }}j_{\nu }}^{\nu
}(\sigma _{n}). 
\]
\end{proposition}

From the properties of the matrix $Q(\alpha )$ (Cor.~\ref{68}) one gets

\begin{proposition}
\label{prop33}
if $d>n-2$, then $\det Q(\alpha )\neq 0$ and the irreducible representations 
$\Phi _{A}^{\alpha }$ described in Th.~\ref{th30} are induced representation $
\ind_{S(n-2)}^{S(n-1)}(\varphi ^{\alpha })$ for the subalgebra $
V_{d}[S(n-1)]\subset $ $A_{n}^{t_{n}}(d),$ so their dimension is equal to $
(n-1)\dim \varphi ^{\alpha }.$ When $d\leq n-2$, then for some $\varphi
^{\alpha }$ it may appear that $\det Q(\alpha )=0$ and consequently the
irreducible representation $\Phi ^{\alpha }$ of $A_{n}^{t_{n}}(d)$ is define
on a subspace of the irreducible representation $\ind_{S(n-2)}^{S(n-1)}(
\varphi ^{\alpha })$.
\end{proposition}

When $\det Q(\alpha )\neq 0$, the equivalent form of the irreducible
representation $\Phi _{A}^{\alpha }$ form Th.~\ref{th30}, in the basis $
\{e_{i}^{a}(\alpha ):a=1,...,n-1,\quad i=1,...,\dim \varphi ^{\alpha }\}$ is
given in the Prop.~\ref{prop48},~\ref{prop50}.

The representations of the algebra $A_{n}^{t_{n}}(d)$ included in the ideal $
S$ are much simpler.

\begin{theorem}
\label{th34}
Each irreducible representation $\psi ^{\nu }$ of the group $S(n-1)$, which
appears in the decomposition of the ideal $S$ given in the Th.~\ref{th30} b) defines
irreducible representations $\Psi ^{\nu }$ of the algebra $A_{n}^{t_{n}}(d)$
in the following way
\[
\Psi ^{\nu }(a)= \left\{ \begin{array}{ll}
0 & \textrm{if $ \ a\in M$,}\\
\psi ^{\nu }(\sigma _{n}) & \textrm{if $ \ a=\sigma _{n}\in S(n-1)$.}
\end{array} \right.
\]
So in this representation the non-invertible element of the ideal $M$ are
represented trivially by zero and therefore we call these representation of
the algebra $A_{n}^{t_{n}}(d)$ semi-trivial. The matrix forms of these
representations \ are simply matrix forms of the irreducible representations
of the group algebra $\mathbb{C}\lbrack S(n-1)]\subset A_{n}^{t_{n}}(d)$ and zero matrices for the elements
of the ideal $M$.
\end{theorem}

\begin{corollary}
\label{col35}
\bigskip All irreducible representations of the algebra $A_{n}^{t_{n}}(d)$
of dimension one are included in the ideal $S.$ In particular, because 
 for the identity representation $\psi ^{id}$%
 of $S(n-1)$ we have $h(\psi ^{id})=1<d:d\geq 2$, the algebra $A_{n}^{t_{n}}(d)$ has a trivial
representation $\Psi^{id}$, in which the elements of the ideal $M$ are represented by
zero and the elements $V_{d}(\sigma ):\sigma \in S(n-1)$ are represented by
number $1$.
\end{corollary}

Let us consider some examples.

\begin{example}
The simplest case is when $n=2$  and $d\geq
2.$ The corresponding algebra $A_{2}^{t_{2}}(d)=\operatorname{span}_{\mathbb{C}
}\{1,V(12)^{t_{2}}:(V(12)^{t_{2}})^{2}=dV(12)^{t_{2}}\}$  is
two-dimensional and commutative. In this simplest case we cannot apply  Th.
31 because the group $S(n-2)=S(2-2)=S(0)$ does not exist. However
due to the simplicity, in particular commutativity of $A_{2}^{t_{2}}(d)$%
  it is easy to derive the structure of the algebra using
elementary calculations, which give%
\[
A_{2}^{t_{2}}(d)=M\oplus S:M=%
\mathbb{C}
V(12)^{t_{2}},\quad S=%
\mathbb{C}
(\text{\noindent
\(\mathds{1}\)}-\frac{1}{d}V(12)^{t_{2}})\Rightarrow A_{2}^{t_{2}}(d)\simeq 
\mathbb{C}
\oplus 
\mathbb{C}
\]%
and %
\[
V(12)^{t_{2}})^{2}=dV(12)^{t_{2}},\quad V(12)^{t_{2}}(\text{\noindent
\(\mathds{1}\)}-\frac{1}{d}%
V(12)^{t_{2}})=0.
\]%
In this simplest all irreducible representations of the algebra $%
A_{2}^{t_{2}}(d)$  are one-dimensional and they are generated by
one-dimensional ideals $M$  and $S$. In the  second
representation the generator  $V(12)^{t_{2}}$  is represented by
the zero operator.
\end{example}

\begin{example}
\label{ex36}
\begin{enumerate}[I)]
\item $n=3$ and $d\geq 3$ all
generating elements $V_{d}^{t_{n}}(\sigma ):\sigma \in S(n)$  are
linearly independent (Th.87, Prop.16) and we have \ $S(3-2)=S(1)=\{\id\}$ and this group has only trivial
irreducible representation $\varphi ^{\id}$. The group $S(3-1)=S(2)$ has two
one-dimensional irreducible representations, trivial $\psi ^{\id}$ and $\sgn$
denoted $\psi ^{s}$ and the induced representation $\ind_{S(1)}^{S(2)}(
\varphi ^{\id})$ has the following decomposition 
\[
\ind_{S(1)}^{S(2)}(\varphi ^{\id})=\psi ^{\id}\oplus \sgn 
\]%
and is a regular representation of the group $S(2).$ In this case the
matrices $Q(\alpha )$ and $Z(\alpha )$ have the form 
\[
Q(\id)=Q(\sgn)=\left( 
\begin{array}{cc}
d & 1 \\ 
1 & d
\end{array}
\right) ,\qquad Z(\id)=\frac{\sqrt{2}}{2}\left( 
\begin{array}{cc}
1 & -1 \\ 
1 & 1
\end{array}
\right) 
\]
and, for $d\geq 2$, $Q(\id)$ is always invertible, because we have $d>n-2=1$
and therefore in the ideal $M$ we have to deal with the induced
presentations of $A_{2}(d)\subset A_{3}^{t_{3}}(d)$ only. The irreducible
representation $\Phi _{A}^{\id}$ of \ the algebra $A_{3}^{t_{3}}(d)$,
included in the ideal $M$, has the following matrix form 
\[
M_{f}^{\alpha }(V^{t_{3}}(13))=\frac{1}{2}\left( 
\begin{array}{cc}
d+1 & -\sqrt{d^{2}-1} \\ 
-\sqrt{d^{2}-1} & d-1%
\end{array}%
\right) ,~M_{f}^{\alpha }(V^{t_{3}}(23))=\frac{1}{2}\left( 
\begin{array}{cc}
d+1 & \sqrt{d^{2}-1} \\ 
\sqrt{d^{2}-1} & d-1%
\end{array}%
\right) 
\]%
\[
M_{f}^{\alpha }(V^{t_{3}}(123))=\frac{1}{2}\left( 
\begin{array}{cc}
d+1 & \sqrt{d^{2}-1} \\ 
-\sqrt{d^{2}-1} & 1-d%
\end{array}%
\right) ,~M_{f}^{\alpha }(V(\sigma _{3}))=\left( 
\begin{array}{cc}
\id(\sigma _{3}) & 0 \\ 
0 & \sgn(\sigma _{3})%
\end{array}%
\right) , 
\]%
where $\sigma _{3}\in S(2).$ The ideal $S$ is in this case two-dimensional
and is a direct sum of two non-isomorphic one-dimensional semi-trivial
representations of algebra $A_{3}^{t_{3}}(d)$, generated by all irreducible
representations of $S(2)$ (because we have $d\geq n-1=2$), which are the
following%
\[
\Psi ^{\id}(\id)=\Psi ^{\id}(12)=1,\qquad \Psi ^{\id}((123)^{t_{3}})=\Psi
^{\id}((132)^{t_{3}})=\Psi ^{\id}((13)^{t_{3}})=\Psi ^{\id}((23)^{t_{3}})=0 
\]%
and 
\[
\Psi ^{s}(\id)=\Psi ^{s}(12)=-1,\qquad \Psi ^{s}((123)^{t_{3}})=\Psi
^{s}((132)^{t_{3}})=\Psi ^{s}((13)^{t_{3}})=\Psi ^{s}(23)^{t_{3}})=0, 
\]%
The algebra $A_{3}^{t_{3}}(d)$ has the following structure%
\[
A_{3}^{t_{3}}(d)\simeq M(2,\mathbb{C})\oplus \mathbb{C}\oplus \mathbb{C},
\]
so it is isomorphic to algebra $
\mathbb{C} \lbrack V_{d}(S(3)]$.
The representation $\Phi _{A}^{\id}$ from the the ideal $M\simeq M(2,\mathbb{C})$, may be written, using Prop.~\ref{prop50}, in an equivalent but simpler form 
\[
\Phi ^{\id}((12))=\left( 
\begin{array}{cc}
0 & 1 \\ 
1 & 0%
\end{array}%
\right) ,\qquad \Phi ^{\id}((132)^{t_{3}})=\left( 
\begin{array}{cc}
1 & d \\ 
0 & 0%
\end{array}%
\right) ,\qquad \Phi ^{\id}((123)^{t_{3}})=\left( 
\begin{array}{cc}
0 & 0 \\ 
d & 1%
\end{array}%
\right) , 
\]%
\[
\Phi ^{\id}(\id)=\left( 
\begin{array}{cc}
1 & 0 \\ 
0 & 1%
\end{array}%
\right) ,\qquad \Phi ^{\id}((13)^{t_{3}})=\left( 
\begin{array}{cc}
0 & 0 \\ 
1 & d%
\end{array}%
\right) ,\qquad \Phi ^{\id}((23)^{t_{3}})=\left( 
\begin{array}{cc}
d & 1 \\ 
0 & 0%
\end{array}%
\right) .\qquad 
\]%
For the group $S(2)$ this representation is the induced representation $
\ind_{S(1)}^{S(2)}(\varphi ^{\id})$ in a non-reduced form.

\item When $n=3,$ $d=2.$  In this case
the irreducible representation $\sgn$  of the group $S(3)$ 
is not included in $A_{3}(2)$, because $h(\sgn)=3>d=2$,
therefore the generating elements $V_{d}^{t_{n}}(\sigma ):\sigma
\in S(n)$ are linearly dependent (Th.87, Prop.16). The matrix $%
Q(id)$ is still invertible because $d=2\geq n-1=2$ so the
ideal $M$ is the same as in case I). The ideal $S$ is
contains only one irreducible representation $\Psi ^{id}$, because %
$h(\psi ^{id})=1<d=2$ and $h(\psi ^{\sgn})=2\nless d=2.$ So %
$S\simeq 
\mathbb{C}
$ and 
\[
A_{3}^{t_{3}}(2)\simeq M(2,%
\mathbb{C}
)\oplus 
\mathbb{C} \simeq
\mathbb{C} \lbrack V_{2}(S(3)].
\]
\end{enumerate}
\end{example}

\bigskip

The next example  is more interesting.

\begin{example}
\label{ex37}
For $n=4$ we have $S(4-2)=S(2)$ and we have two $one-$dimensional
representations of the group $S(2)$ $\varphi ^{\id},$ $\varphi ^{\sgn}$ and the
group $S(4-1)=S(3)$ has three non-trivial irreducible representations,
characterized by partitions $\nu _{\id}=(3),$ $\nu _{2}=(2,1)$ and $\nu
_{s}=(1,1,1)$. The induced representations $\ind_{S(2)}^{S(3)}(\varphi
^{\alpha })$, where $\alpha =\id, \sgn$ are irreducible representations of $%
S(2)$ have the following structure 
\[
\ind_{S(2)}^{S(3)}(\varphi ^{\id})=\psi ^{\alpha _{2}}\oplus \psi ^{\id},\qquad
\ind_{S(2)}^{S(3)}(\varphi ^{\sgn})=\psi ^{\alpha _{2}}\oplus \psi ^{\sgn} 
\]%
and we chose the following matrix representation for $\psi ^{\nu _{2}}$%
\[
\psi ^{\nu _{2}}(12)=\left( 
\begin{array}{cc}
0 & 1 \\ 
1 & 0%
\end{array}%
\right) ,\quad \psi ^{\nu _{2}}(13)=\left( 
\begin{array}{cc}
0 & \varepsilon \\ 
\varepsilon ^{-1} & 0%
\end{array}%
\right) ,\quad \psi ^{\nu _{2}}(23)=\left( 
\begin{array}{cc}
0 & \varepsilon ^{-1} \\ 
\varepsilon & 0%
\end{array}%
\right) ,\quad 
\]%
\[
\psi ^{\nu _{2}}(123)=\left( 
\begin{array}{cc}
\varepsilon & 0 \\ 
0 & \varepsilon ^{-1}%
\end{array}%
\right) ,\quad \psi ^{\nu _{2}}(132)=\left( 
\begin{array}{cc}
\varepsilon ^{-1} & 0 \\ 
0 & \varepsilon%
\end{array}%
\right) ,\quad 
\]%
where $\varepsilon ^{3}=1.$ In case of the algebra $A_{4}^{t_{4}}(d)$, its
structure and the structure of its irreducible representations depends on
the value of the dimension $d.$ We have two cases.

\begin{enumerate}[I)]
\item When $n=4, d \geq 4$, then the condition $d>n-2=2$ is satisfied and the
irreducible representations $\Phi _{A}^{\alpha }$, from the ideal $M,$ are
induced representations $\ind_{S(2)}^{S(3)}(\varphi ^{\alpha })$ for the
subalgebra $A_{3}(d)\subset A_{4}^{t_{4}}(d)$, where $\alpha =\id,$ $\sgn$ are
irreducible representations of $S(2).$ Using the explicit form of the
representation $\psi ^{\nu _{2}}$ of \ $S(3)$ we get that the matrices $%
Q(\id) $ and $Z(\id)$ are of the form 
\[
Q(\id)=\left( 
\begin{array}{ccc}
d & 1 & 1 \\ 
1 & d & 1 \\ 
1 & 1 & d%
\end{array}%
\right) ,\quad Z(\id)=\frac{1}{\sqrt{3}}\left( 
\begin{array}{ccc}
\varepsilon & \varepsilon ^{2} & 1 \\ 
\varepsilon ^{2} & \varepsilon & 1 \\ 
1 & 1 & 1%
\end{array}%
\right) . 
\]%
The irreducible representation $\Phi _{A}^{\id}$ of \ the algebra $%
A_{4}^{t_{4}}(d)$, included in the ideal $M$, has the following matrix form 
\[
M_{f}^{\id}(V^{t_{4}}(14))=\frac{1}{3}D^{\id}\left( 
\begin{array}{ccc}
1 & \varepsilon & \varepsilon ^{2} \\ 
\varepsilon ^{2} & 1 & \varepsilon \\ 
\varepsilon & \varepsilon ^{2} & 1%
\end{array}%
\right) D^{\id},\quad M_{f}^{id}(V^{t_{4}}(24))=\frac{1}{3}D^{\id}\left( 
\begin{array}{ccc}
1 & \varepsilon ^{2} & \varepsilon \\ 
\varepsilon & 1 & \varepsilon ^{2} \\ 
\varepsilon ^{2} & \varepsilon & 1%
\end{array}%
\right) D^{\id}, 
\]%
\[
M_{f}^{\id}(V^{t_{4}}(34))=\frac{1}{3}D^{\id}\left( 
\begin{array}{ccc}
1 & 1 & 1 \\ 
1 & 1 & 1 \\ 
1 & 1 & 1%
\end{array}%
\right) D^{\id},\quad M_{f}^{\id}(V(\sigma _{4}))=\left( 
\begin{array}{cc}
\psi ^{2}(\sigma _{4}) & 0 \\ 
0 & 1%
\end{array}%
\right) ,\quad \sigma _{4}\in S(3), 
\]%
where 
\[
D^{\id}=\left( 
\begin{array}{ccc}
\sqrt{d-1} & 0 & 0 \\ 
0 & \sqrt{d-1} & 0 \\ 
0 & 0 & \sqrt{d+2}%
\end{array}%
\right) 
\]%
In the basis from Prop.~\ref{prop50} this irreducible representation of $%
A_{4}^{t_{4}}(d)$ looks simpler 
\[
M_{e}^{\id}(V^{t_{4}}(14))=\left( 
\begin{array}{ccc}
d & 1 & 1 \\ 
0 & 0 & 0 \\ 
0 & 0 & 0%
\end{array}%
\right) ,\quad M_{e}^{\id}(V^{t_{4}}(24))=\left( 
\begin{array}{ccc}
0 & 0 & 0 \\ 
1 & d & 1 \\ 
0 & 0 & 0%
\end{array}%
\right) , 
\]%
\[
M_{e}^{\id}(V^{t_{4}}(34))=\left( 
\begin{array}{ccc}
0 & 0 & 0 \\ 
0 & 0 & 0 \\ 
1 & 1 & d%
\end{array}%
\right) ,\quad M_{e}^{\id}(V(\sigma _{4})_{ij}=\delta _{i\sigma (j)},\quad
\sigma _{4}\in S(3) 
\]%
So for $A_{3}(d)\subset A_{4}^{t_{4}}(d)$ this a natural representation.
Similarly, for the irreducible representation $\Phi _{A}^{\sgn}$ we get that 
\[
Q(\sgn)=\left( 
\begin{array}{ccc}
d & -1 & 1 \\ 
-1 & d & 1 \\ 
1 & 1 & d%
\end{array}%
\right) ,\quad Z(\sgn)=\frac{1}{\sqrt{3}}\left( 
\begin{array}{ccc}
-\varepsilon & -\varepsilon ^{-1} & -1 \\ 
-\varepsilon ^{-1} & -\varepsilon & -1 \\ 
1 & 1 & 1%
\end{array}%
\right) 
\]%
The irreducible representation $\Phi _{A}^{\sgn}$ of \ the algebra $%
A_{4}^{t_{4}}(d)$, included in the ideal $M$, has the following matrix form 
\[
M_{f}^{\sgn}(V^{t_{4}}(14))=\frac{1}{3}D^{s}\left( 
\begin{array}{ccc}
1 & \varepsilon ^{2} & \varepsilon \\ 
\varepsilon & 1 & \varepsilon ^{2} \\ 
\varepsilon ^{2} & \varepsilon & 1%
\end{array}%
\right) D^{s},\quad M_{f}^{\sgn}(V^{t_{4}}(24))=\frac{1}{3}D^{s}\left( 
\begin{array}{ccc}
1 & \varepsilon & \varepsilon ^{2} \\ 
\varepsilon ^{2} & 1 & \varepsilon \\ 
\varepsilon & \varepsilon ^{2} & 1%
\end{array}%
\right) D^{s}, 
\]%
\[
M_{f}^{\sgn}(V^{t_{4}}(34))=\frac{1}{3}D^{s}\left( 
\begin{array}{ccc}
1 & 1 & 1 \\ 
1 & 1 & 1 \\ 
1 & 1 & 1%
\end{array}%
\right) D^{s},\quad M_{f}^{\sgn}(V(\sigma _{4}))=\left( 
\begin{array}{cc}
\psi ^{2}(\sigma _{4}) & 0 \\ 
0 & \sgn(\sigma _{4})%
\end{array}%
\right) , 
\]%
where 
\[
D^{s}=\left( 
\begin{array}{ccc}
\sqrt{d+1} & 0 & 0 \\ 
0 & \sqrt{d+1} & 0 \\ 
0 & 0 & \sqrt{d-2}%
\end{array}%
\right) 
\]%
Again, in the basis from Prop.~\ref{prop50} the irreducible representation $\Phi
_{A}^{\sgn}$ of $A_{4}^{t_{4}}(d)$ looks simpler 
\[
M_{e}^{\sgn}(V^{t_{4}}(14))=\left( 
\begin{array}{ccc}
d & -1 & 1 \\ 
0 & 0 & 0 \\ 
0 & 0 & 0%
\end{array}%
\right) ,\quad M_{e}^{\sgn}(V^{t_{4}}(24))=\left( 
\begin{array}{ccc}
0 & 0 & 0 \\ 
-1 & d & 1 \\ 
0 & 0 & 0%
\end{array}%
\right) , 
\]%
\[
M_{e}^{\sgn}(V^{t_{4}}(34))=\left( 
\begin{array}{ccc}
0 & 0 & 0 \\ 
0 & 0 & 0 \\ 
1 & 1 & d%
\end{array}%
\right) ,\quad M_{e}^{\sgn}(V(\sigma _{4})_{ij}=\delta _{i\sigma
(j)}\sgn[(i3)\sigma _{4}(j3)], 
\]%
The ideal $M$ has the following structure 
\[
M=M(3,\mathbb{C})\oplus M(3,\mathbb{C}),
\]%
where each induced representation $\ind_{S(2)}^{S(3)}(\varphi ^{\alpha })$ is
included with the multiplicities equal to their dimensions.The ideal $S$, is a
direct sum of three non-isomorphic semi-trivial representations $\Psi ^{\nu
} $ of the algebra $A_{4}^{t_{4}}(d)$ (also with multiplicities equal to their
dimensions), generated by all irreducible representations $\psi ^{\nu }$ of $%
S(3)$ ($\nu _{\id}=(3),$ $\nu _{2}=(2,1)$ and $\nu _{s}=(1,1,1))$, because in
this case we have $d\geq n=4$ (see Theorem~\ref{th29}b). It means that, in these semi-trivial
representations the natural generators $V(\sigma _{n}):\sigma _{n}\in S(n-1)$
are represented by operators $\psi ^{\nu }(\sigma _{n})$, whereas the
elements of the ideal $M$ are represented by zero operator. Therefore the
ideal $S$ has the following structure 
\[
S=M(2,\mathbb{C})\oplus \mathbb{C} \oplus \mathbb{C}. 
\]
\item $n=4,$ $d=3$. In this case the irreducible
representation $\sgn$ of the group $S(4)$ is not included
in $A_{4}(3)$, because $h(\sgn)=4>d=3$, so from (Th.87,
Prop.16) we get that  the generating elements $%
V_{d}^{t_{n}}(\sigma ):\sigma \in S(4)$ are linearly dependent. The
matrices $Q(id)$ and $Q(\sgn)$ are invertible, because $%
d=3\geq n-1=3,$ therefore the ideal $M$ has the same form
as in previous case I). The ideal $S$ is contains the irreducible
representations $\Psi ^{id}$ and $\Psi ^{\nu _{2}}$, $%
h(\Psi ^{id}),$ $h(\Psi ^{\nu _{2}})<3,$ so $S\simeq M(2,%
\mathbb{C}
)\oplus 
\mathbb{C}
.$

\item In the case $d=2,$ the irreducible
representations $sgn$ and $(2,1,1)$ of the group $S(4)$%
 are not included in $A_{4}(2)$, because $h(\sgn)=4,$ $%
h(2,1,1)=3>d=2$, therefore, again from Th.87, Prop.16 we
get that the generating elements $V_{d}^{t_{n}}(\sigma ):\sigma
\in S(4)$ are linearly dependent. The ideal $M$ in 
algebra $A_{4}^{t_{4}}(2)$ has another structure of its irreducible
representations, because we have $d=n-2=2$ and from Prop.34 and Cor.67 it
follows that, in this case the irreducible irreducible representation $\Phi
_{A}^{id}$ is the induced representation $ind_{S(2)}^{S(3)}(\varphi ^{id})$
for the subalgebra $A_{3}(d)$ but the irreducible irreducible representation 
$\Phi _{A}^{\sgn}$ is defined only on two-dimensional subspace of the
representation $ind_{S(2)}^{S(3)}(\varphi ^{\sgn})$ (because $\det Q(\sgn)=0).$
The representation $\Phi _{A}^{id}$ $\ $has, in this case the same form as
for the case I) but we have to set $d=2$ and the irreducible representation $%
\Phi _{A}^{\sgn}$ of $A_{4}^{t_{4}}(d)$ has the following matrix form 
\[
M_{f}^{\sgn}(V^{t_{4}}(14))=\left( 
\begin{array}{cc}
1 & \varepsilon ^{2} \\ 
\varepsilon  & 1%
\end{array}%
\right) ,\quad M_{f}^{\sgn}(V^{t_{4}}(24))=\left( 
\begin{array}{cc}
1 & \varepsilon  \\ 
\varepsilon ^{2} & 1%
\end{array}%
\right) ,
\]%
\[
M_{f}^{\sgn}(V^{t_{4}}(34))=\left( 
\begin{array}{cc}
1 & 1 \\ 
1 & 1%
\end{array}%
\right) ,\quad M_{f}^{\sgn}(V(\sigma _{4}))=\psi ^{2}(\sigma _{4}),\quad
\sigma _{4}\in S(3)
\]%
In this case we can not use the Prop.49 because $\det Q(\sgn)=0$ for $d=2.$
We have also 
\[
M=M(3,%
\mathbb{C}
)\oplus M(2,%
\mathbb{C}
).
\]%
Similarly, from Th.30b and 35 it follows
that, in this case the representations $\Psi ^{\sgn}$ and $\Psi
^{\nu _{2}}$of the group $S(3)$ are not included in the
ideal $S$, because $h(\sgn)=3,$ $h(\nu _{2})=2\nless d=2,$%
 Therefore the ideal $S$ contains only one irreducible
representation $\Psi ^{id}$ and has the following structure
\[
S\simeq 
\mathbb{C}
.
\]
\end{enumerate}
\end{example}

The results obtained in these examples were obtained, in an equivalent form,
in~\cite{Studzinski1} were different method has been used.

\section{\protect\bigskip Construction of the irreducible representations of
the algebra $A_{n}^{t_{n}}(d).$}
\label{construction1}

\bigskip

In this section we will prove our main results stated in Theorems
30, 31 and 35 of previous section. In the proof we will construct all
representations of the algebra $A_{n}^{t_{n}}(d).$ Because the algebra $%
A_{n}^{t_{n}}(d)$ is semi-simple then, from the Theorem 5 and Proposition 7
we get that all irreducible representations of the algebra $A_{n}^{t_{n}}(d)$
are included, as a left minimal ideals, in the left regular representation
of this algebra. More precisely these left minimal ideals are included in
the matrix ideals of $A_{n}^{t_{n}}(d)$, so in fact we have to look for the
matrix ideals of the algebra $A_{n}^{t_{n}}(d).$ This a purely algebraical
approach, alternative to representation approach described in~\cite{Studzinski1}, is based on he properties of the multiplication in the
algebra $A_{n}^{t_{n}}(d)$ described in Th.19. First we we will construct the
matrix ideals of $M$ and all irreducible representations of the algebra $%
A_{n}^{t_{n}}(d)$ included in the ideal $M$, next all irreducible
representations included in the complementary ideal $S,$ which is closely
related with the algebra $%
\mathbb{C}
\lbrack S(n-1)].$

\subsection{ Matrix ideals and left minimal ideals of the ideal $M$ of the
algebra $A_{n}^{t_{n}}(d).$}
\label{subconstr1}

\bigskip We define new generating elements of the ideal $M$ in the following
way

\begin{definition}
Let $\varphi ^{\alpha }$, $\alpha =1,....,k$ be all inequivalent irreps of
the group $S(n-2)$ which are supposed to be unitary and $w^{\alpha
}=\dim \varphi ^{\alpha }.$ For any $\varphi ^{\alpha }$, $\alpha =1,....,k$
we define 
\[
u_{ij}^{ab}(\alpha )=\frac{w^{\alpha }}{(n-2)!}V(an)^{t_{n}}\sum_{\sigma \in
S(n-2)}\varphi _{ji}^{\alpha }(\sigma ^{-1})V[(a \ n-1)(\sigma )(b \ n-1)] 
\]
\[
=V(an)^{t_{n}}V[(a \ n-1)]V[E_{ij}^{\alpha }]V[(b \ -1)] 
\]
where $a,b=1,..,n-1$ and $E_{ij}^{\alpha }$ are matrix
operators of the representation $\varphi ^{\alpha }$ defined in the Def. 69,
Appendix C and $i,j=1,..,w^{\alpha }$ are indices of the matrix form of $%
\varphi ^{\alpha }.$
\end{definition}

\begin{remark}
\label{rem39}
\bigskip Here, a natural question arises, why the group $S(n-2)$ and its
representations appears in this Definition. The group $S(n-2)$ arises in a
more natural way in the studies of irreducible representations of the
algebra $A_{n}^{t_{n}}(d))$ in the representation approach where one
consider the properties of the operators of the algebra $A_{n}^{t_{n}}(d)$
acting on a given basis of its natural representation space $\mathcal{(%
\mathbb{C}
}^{d})^{\otimes n}$~\cite{Studzinski1}. Using however, the results of this
section, one can show that the ideal $M$ is a direct sum of subalgebras
homomorphic to the group algebras $%
\mathbb{C}
\lbrack S(n-2)]$ (but we omit this derivation) which allows to construct
left ideals in $M$ (i.e. representations of the algebra $A_{n}^{t_{n}}(d))$
using irreducible representations of the group $S(n-2).$
\end{remark}

For a given irreducible representation $\varphi ^{\alpha }$ of the group $%
S(n-2)$ we get a set of $((n-1)\dim \varphi ^{\alpha })^{2}$ elements $%
\{u_{ij}^{ab}(\alpha )\}$ which are either all non-zero vectors or all of
them are zero. In fact we have

\begin{proposition}
\label{prop40}
\begin{enumerate}[a)]
\item For a given irreducible representation $\varphi ^{\alpha }$ of
the group $S(n-2)$ the \ operators $u_{ij}^{ab}(\alpha ):a,b=1,...,n-1,\quad
\quad i,j=1,..,w^{\alpha }$ \ are different from zero iff the the
irreducible representation \ $\varphi ^{\alpha }$ of the group $S(n-2)$
appears in the permutation representation $V_{d}[S(n-2)]$ of the group $%
S(n-2)$.

\item The vectors $\{u_{ij}^{ab}(\alpha )\}$ span the ideal $M$ i.e. we have 
\[
M=\Span_{%
\mathbb{C}
}\{u_{ij}^{ab}(\alpha ):a,b=1,...,n-1,\quad \varphi ^{\alpha }\in
V_{d}[S(n-2)],\quad i,j=1,..,w^{\alpha }\} 
\]
\end{enumerate}
\end{proposition}

\begin{proof}
\begin{enumerate}[a)]
\item The representation $V:$ $S(n)$ $\rightarrow \Hom(\mathcal{(%
\mathbb{C}
}^{d})^{\otimes n})$ of the group $S(n)$ is also, in a natural way, the
representation of its subgroup $S(n-2)$ i.e. we have also $V_{d}:$ $S(n-2)$ $%
\rightarrow \Hom(\mathcal{(%
\mathbb{C}
}^{d})^{\otimes n})$. From the statements of the Appendix~\ref{AppC} we know that the
matrix operators%
\[
E_{ij}^{\alpha }=\frac{w^{\alpha }}{(n-2)!}\sum_{\sigma \in S(n-2)}\varphi
_{ji}^{\alpha }(\sigma ^{-1})V_{d}[\sigma ]\in \Hom(\mathcal{(%
\mathbb{C}
}^{d})^{\otimes n}) 
\]%
\bigskip of the representation $\varphi ^{\alpha }$ of the subgroup $S(n-2)$
are nonzero iff the irreducible representation $\varphi ^{\alpha }$ belongs
to the representation $V:$ $S(n-2)$ $\rightarrow \Hom(\mathcal{(%
\mathbb{C}
}^{d})^{\otimes n}).$ From the invertibility of the operators $V(an)$ and $%
V[(an-1)]$ (in the representation $V:$ $S(n)$ $\rightarrow \Hom(\mathcal{(%
\mathbb{C}
}^{d})^{\otimes n}))$ we get that for any $a,b=1,...,n-1$ the elements 
\[
\widetilde{u}_{ij}^{ab}(\alpha )=V(an)V[(a \ n-1)]V[E_{ij}^{\alpha
}]V[(b \ n-1)]\in A_{n}(d) 
\]%
are nonzero iff $\varphi ^{\alpha }$ belongs to the representation $V:$ $%
S(n-2)$ $\rightarrow \Hom(\mathcal{(\mathbb{C}}^{d})^{\otimes n}).$ Now the statement of the Proposition follows from the
invertibility of the partial transpose $t_{n}$ (which is a linear
transformation) and from properties of multiplication in the algebra $%
A_{n}^{t_{n}}(d))$ Th.~\ref{th19}.

\item The algebra $A_{n}(d)$ , which homomorphic with the group algebra $\mathbb{C}
\lbrack S(n)],$ is a sum of two subspaces 
\[
 A_{n}(d)=\widetilde{M}+V[%
\mathbb{C}
\lbrack S(n-1)]
\]%
where 
\[
\widetilde{M}=\Span_{%
\mathbb{C}
}\{V(\sigma _{ab}):a,b\neq n\}. 
\]%
The subspace $\widetilde{M}$ has the following form
\[
\widetilde{M}=\sum_{a,b=1,...,n-1}S_{ab}:S_{ab}=span_{%
\mathbb{C}
}\{V(\sigma ):\sigma \in S(n)\wedge \sigma =\sigma _{ab}\}. 
\]%
It is easy to check that for a given pair $(a,b),$ $a,b=1,...,n-1$ we have 
\[
S_{ab}=V(an)V[(a \ n-1)]V_{d}[S(n-2)]V[(b \ n-1)], 
\]%
so taking into account (see Appendix~\ref{AppC}) that 
\[
V_{d}[S(n-2)]=\bigoplus _{\varphi ^{\alpha }\in V_{d}[S(n-2)]} \Span_{%
\mathbb{C}
}\{E_{ij}^{\alpha }\} 
\]%
we get that 
\[
\widetilde{M}=\Span_{%
\mathbb{C}
}\{\widetilde{u}_{ij}^{ab}(\alpha ):a,b=1,...,n-1,\quad ,\varphi ^{\alpha
}\in V_{d}[S(n-2)],\quad i,j=1,..,w^{\alpha }\} 
\]%
and again from the invertibility of the partial transpose we get the
statement b) of the Proposition.
\end{enumerate}
\end{proof}

\bigskip

Similarly as the generating elements $V(\sigma )^{t_{n}}:\sigma \in S(n)$
the vectors $u_{ij}^{ab}(\alpha ),$ in general, need not to be linearly
independent which depends on the the dimension parameter $d.$ Arguing
however, in a similar way as in the above proof and using Prop.~\ref{prop14}, one gets

\begin{proposition}
\label{prop41}
a) \bigskip The vectors $\{u_{ij}^{ab}(\alpha )\}$, in general, need not to
bee linearly independent but for a given pair $(a,b),$ $a,b=1,...,n-1$ and
given irreducible representation $\varphi ^{\alpha }$ $\in $ $V_{d}[$ $%
S(n-2)]$ the set of $(w^{\alpha })^{2}$ nonzero vectors $u_{ij}^{ab}(\alpha
),\quad i,j=1,..,w^{\alpha }$ is linearly independent.

b) When $d\geq n$ then the both sets of vectors $\{V(\sigma )^{t_{n}}:\sigma
\in S(n)\}$ and $\{u_{ij}^{ab}(\alpha ):a,b=1,...,n-1,\quad \alpha
=1,....,k,\quad i,j=1,..,w^{\alpha }\}$ are linearly independent and they
form a bases of the ideal $M.$
\end{proposition}

Using multiplication rule for the algebra $A_{n}^{t_{n}}(d))$ (Th.~\ref{th19}), one
can derive the following properties of the elements $\{u_{ij}^{ab}(\alpha
)\}$ of $M\subset A_{n}^{t_{n}}(d)$

\begin{proposition}
\label{prop42}
Suppose that $\varphi ^{\alpha }\in V[S(n-2)]$, i.e. $u_{ij}^{pq}(\alpha
)\neq 0$, then for any $V(\sigma _{ab})^{t_n }\in M$ we have 
\be
\label{42eq1}
V(\sigma _{ab})^{t_n }u_{ij}^{pq}(\alpha )=d^{\delta
_{ap}}\sum_{k=1}^{w^{\alpha }}\varphi _{ki}^{\alpha }[(b \ n-1)\widehat{\sigma }%
_{ab}(ap)(p \ n-1)]u_{kj}^{bq}(\alpha ),
\ee
where $\widehat{\sigma }_{ab}=(bn)\sigma _{ab}=\sigma _{ab}(an)\in S(n-1),$ $%
a,b\neq n$, in particular 
\[
V(an)^{t_n }u_{ij}^{pq}(\alpha )=d^{\delta _{ap}}\sum_{k=1}^{w^{\alpha
}}\varphi _{ki}^{\alpha }[(a \ n-1)(ap)(p \ n-1)]u_{kj}^{aq}(\alpha ), 
\]
and for any $\sigma _{n}\in S(n-1)$ 
\be
\label{42eq2}
V(\sigma _{n})u_{ij}^{pq}(\alpha )=\sum_{k=1}^{w^{\alpha }}\varphi
_{ki}^{\alpha }[(\sigma _{n}(p) \ n-1)\sigma _{n}(p \ n-1)]u_{kj}^{\sigma
_{n}(p)q}(\alpha ),
\ee
the multiplication rule for these vectors is the following%
\be
\label{42eq3}
u_{ij}^{ab}(\alpha )u_{kl}^{pq}(\beta )=\delta _{\alpha \beta }d^{\delta
_{bp}}\varphi _{jk}^{\alpha }[(b \ n-1)(bp)(p \ n-1)]u_{il}^{aq}(\alpha ). 
\ee
In particular one has 
\[
u_{ii}^{aa}(\alpha )u_{ii}^{aa}(\beta )=\delta _{\alpha \beta
}du_{ii}^{aa}(\alpha ), 
\]%
i.e. the vectors $u_{ii}^{aa}(\alpha )$ \ are essential projectors.
\end{proposition}

From the transformation law of the natural generators (Eq.1, Eq.2) we see
that the left action of an arbitrary element $V(\sigma )^{t_n}\in A$ on
the elements $u_{ij}^{cq}(\alpha )$ changes only the first upper and first
lower indices in $u_{ij}^{cq}(\alpha ),$ whereas the second indices remain
unchanged. Therefore we may formulate

\begin{corollary}
\label{col43}
For a given partition $\alpha $ such that $\varphi ^{\alpha }\in V[S(n-2)]$
and for arbitrary fixed values of $b\in \{1,...,n-1\}$ and $j\in
\{1,...,w^{\alpha }\}$ a linear space 
\[
U_{j}^{b}(\alpha )\equiv \Span_{%
\mathbb{C}
}\{u_{ij}^{ab}(\alpha ):a=1,...,n-1,\quad i=1,...,w^{\alpha }\}\subset M 
\]%
is a left submodule of the algebra $A_{n}^{t_{n}}(d)$ in $M$ i.e. it is a
representation of the algebra $A_{n}^{t_{n}}(d)$ which appears in the
regular representation of the algebra $A_{n}^{t_{n}}(d)$ with the
multiplicity $(n-1)\dim \varphi ^{\alpha }$ in $M.$ It is clear also that
for $\alpha $ such that $\varphi ^{\alpha }\in V[S(n-2)]$ the subspace 
\[
U(\alpha )=\bigoplus _{a=1,i=1}^{n-1,\dim \varphi ^{\alpha }}U_{i}^{a}(\alpha
)=\Span_{%
\mathbb{C}
}\{u_{ij}^{ab}(\alpha ):a,b=1,...,n-1,\quad i,j=1,...,w^{\alpha }\} 
\]%
is an ideal of \ $M$ $\subset A_{n}^{t_{n}}(d)$, generated by $((n-1)\dim
\varphi ^{\alpha })^{2}$ elements. We have also the following decomposition
of the ideal $M$ into a direct sum of ideals 
\[
M=\bigoplus _{\alpha }U(\alpha), 
\]%
where the sum is over $a:\varphi ^{\alpha }\in V[S(n-2)].$
\end{corollary}

In fact the representation $U_{i}^{a}(\alpha )$ of the algebra $%
A_{n}^{t_{n}}(d)$ is generated by the irreducible representation $\varphi
^{\alpha }$ of the group $S(n-2)$ . We see also that each irreducible
representation $\varphi ^{\alpha },$ $\alpha =1,....,k$ of the group $%
S(n-2), $ such that $\varphi ^{\alpha }\in V[S(n-2)],$ generates $(n-1)\dim
\varphi ^{\alpha }$ isomorphic representations of the algebra $%
A_{n}^{t_{n}}(d)$, indexed by $a\in \{1,...,n-1\}$ and $i\in
\{1,...,w^{\alpha }\}.$ This isomorphism follows from the fact that the
transformation laws (Eq.~\eqref{42eq1}, Eq.~\eqref{42eq2}) in Prop.~\ref{prop42} does not depend on the indices $q$
and $j.$ We will show that these representations are in fact, all
irreducible representations of the algebra $A_{n}^{t_{n}}(d).$

Obviously the representation $U_{j}^{b}(\alpha )$ of the algebra $%
A_{n}^{t_{n}}(d)$ is a also representation of its subalgbera $%
\mathbb{C}
\lbrack S(n-1)]$ and consequently of the group $S(n-1).$ From the
transformation rule (Eq.~\eqref{42eq2}), for the elements $V(\sigma _{n})\in $ $V[S(n-1)]$
one deduce

\begin{theorem}
\label{th44}
Suppose that the vectors $\{u_{ij}^{ab}(\alpha )\}$ are linearly independent
(so they form a basis of $U(\alpha )$), then for any $b=1,...,n-1$ and $%
j=1,..,\dim \varphi ^{\alpha }$ the representation $U_{j}^{b}(\alpha )$ \ of
the algebra $A_{n}^{t_{n}}(d)$ from Cor.~\ref{col43}, as a reducible representation of
the group $S(n-1)\subset A_{n}^{t_{n}}(d),$ is in fact the representation of $%
S(n-1)$ induced by the irrep $\varphi ^{\alpha }$ of the subgroup $%
S(n-2)\subset S(n-1)$, i.e. $U_{j}^{b}(\alpha )$ is the representation space of $%
\ind_{S(n-2)}^{S(n-1)}(\varphi ^{\alpha }).$ for any $b=1,...,n-1$, $j=1,..,\dim \varphi ^{\alpha }$.
\end{theorem}

\begin{proof}
In our case $G=S(n-1)\supset S(n-2)=H$ and we chose the transpositions $%
(an-1),$ $a=1,...,n-1$ as the natural representatives of the left cosets $%
S(n-1)/S(n-2).$ Then taking the representation $\varphi ^{\alpha }$ of the
subgroup $S(n-2)$, using the Definition~\ref{def8}, we get 
\[
\forall \sigma \in S(n-1)\quad \pi _{ai,bj}(\sigma )=\delta _{a\sigma
(b)}\varphi _{ij}^{\alpha }[(\sigma (b) \ n-1)\sigma (b \ n-1)] 
\]%
which is exactly the block matrix representing the elements $\sigma =\sigma
_{n}\in S(n-1)\subset A_{n}^{t_{n}}(d)$ in the irreducible representation $%
U_{s}^{x}(\alpha )$ \ of the algebra $A_{n}^{t_{n}}(d)$ given in the
Proposition~\ref{prop42}$.$
\end{proof}

The multiplication rule for the elements $\{u_{ij}^{ab}(\alpha )\}$ is
similar to the matrix multiplication and in fact, it is the matrix
multiplication when the representation $\varphi ^{\alpha }$ is the
identity representation of the group $S(n-2).$ When the representation $%
\varphi ^{\alpha }$ is not the trivial representation then the properties of
the multiplication rule for the elements $\{u_{ij}^{ab}(\alpha )\}$ depends
on the properties of the matrix $Q(\alpha )$ (defined in Def.~\ref{def28}), which
appears in the multiplication law of the elements $\{u_{ij}^{ab}(\alpha )\}$
(Prop.~\ref{prop42} Eq.~\eqref{42eq3}).

\bigskip The properties and the structure of the matrices $Q(\alpha )$ are
described in the \ Appendix~\ref{AppA}. The multiplication rule for the elements $%
\{u_{ij}^{ab}(\alpha )\}$ in the ideal $U(\alpha )$ (eq.~\eqref{42eq3}), may be written
now 
\[
u_{ij}^{ab}(\alpha )u_{kl}^{pq}(\beta )=\delta _{\alpha \beta
}Q_{ik}^{bp}(\alpha )u_{il}^{aq}(\alpha ). 
\]

\bigskip

In the Appendix B we derived a basic properties of the algebras with
multiplication of this type. The algebra 
\[
U(\alpha )=\Span_{%
\mathbb{C}
}\{u_{ij}^{ab}(\alpha ):a,b=1,...,n-1,\quad i,j=1,...,w^{\alpha }\} 
\]%
is semisimple (because it is $C^{\ast }-$algebra, Th.4) and its properties
depends strongly on the properties of the matrix $Q(\alpha )$. In particular
the dimension of the algebra $U(\alpha )$ depend on the rank of the matrix $%
Q(\alpha ).$

From the results derived in the Appendices A and B it follows that if $\det
Q(\alpha )\neq 0$, then all vectors $\{u_{ij}^{ab}(\alpha )\}$ are linearly
independent (so they form a basis of $U(\alpha )$) and the algebra $U(\alpha
)$ is isomorphic with the matrix algebra $M((n-1)\dim \varphi ^{\alpha },%
\mathbb{C}
)$ (see Th.~\ref{th74} and Th.~\ref{th77} in Appendix~\ref{AppB})$.$ The case when $\det Q(\alpha )=0$
is more complicated, because then the vectors $\{u_{ij}^{ab}(\alpha )\}$ are
linearly dependent and we have to construct a basis of the algebra $%
U(\alpha ).$ The universal construction of a reduced basis, which applies to
the both cases is described in the Appendix~\ref{AppB} (Th.~\ref{th77}). The basis constructed
in this theorem we will call the reduced basis of the ideal $U(\alpha )$
because firstly it is constructed using the matrix $Z(\alpha ),$ that reduces
the induced representation $\ind_{S(n-2)}^{S(n-1)}(\varphi ^{\alpha })$ to
the direct sum of irreducible representations and secondly, in case when $%
\det Q(\alpha )=0$ in its construction we reduce the set of linearly
dependent generators to linearly independent one.

Now, using Th.~\ref{th77} from the Appendix~\ref{AppB} we will construct the reduced basis in
the algebra $U(\alpha )$ when $\varphi ^{\alpha }\in V_{d}[S(n-2)].$ The
construction of the basis is based on the diagonalization of the matrix $%
Q(\alpha )$ appearing in the multiplication law of the algebra $U(\alpha )$
(Eq.~\eqref{42eq3}), which is realized by the the following similarity transformation (see
Th.~\ref{th58}, Eq.~\eqref{A2} Appendix~\ref{AppA}) 
\[
\sum_{ak}\sum_{bl}z^{\dagger}(\alpha )_{j_{\rho }k}^{\rho a}Q(\alpha
)_{kl}^{ab}z(\alpha )_{lj_{\mu }}^{b\mu }=\delta ^{\rho \mu }\delta
_{j_{\rho }j_{\mu }}\lambda _{\mu }, 
\]%
where the unitary matrix $Z(\alpha )=($ $z(\alpha )_{lj_{\mu }}^{b\mu })$
reduces the induced representation $\ind_{S(n-2)}^{S(n-1)}(\varphi ^{\alpha
}) $ to the sum $\oplus _{\lambda }\psi ^{\lambda }$ of irreducible
representations of $S(n-1)$ (Prop.~\ref{prop27}).

Now using the matrix $Z(\alpha )$, according to the Th.~\ref{th77} in the Appendix~\ref{AppB}
we define new elements in the ideal $U(\alpha )$ in the following way 
\be
\label{eq4}
y_{j_{\nu }j_{\mu }}^{\nu \mu }(\alpha )=\sum_{ai,bk}(z^{-1})_{j_{\mu
}i}^{\mu a}u_{ki}^{ba}(\alpha )z_{kj_{\nu }}^{b\nu },
\ee
where the indices $\mu ,\nu $ labels the irreducible representations $\psi
^{\mu },\psi ^{\upsilon }$ of the group $S(n-1),$ which appears in the
reduction of the induced representation $\Phi ^{\alpha
}=\ind_{S(n-2)}^{S(n-1)}\varphi ^{\alpha }$ and $j_{\nu \text{ }}$are matrix
indices of $\psi ^{\upsilon }$ (see Appendix~\ref{AppA}). \ From the Th.~\ref{th77} in the
Appendix~\ref{AppB} it follows, that if $\det Q(\alpha )=0,$ which means that for one
(and only one , up to the multiplicity, Cor.~\ref{col67} App.~\ref{AppA}) irreducible
representation $\psi ^{\theta }$ of $S(n-1)$ appearing in $%
\ind_{S(n-2)}^{S(n-1)}(\varphi ^{\alpha })=\bigoplus _{\lambda }\psi ^{\lambda }$
the corresponding eigenvalue $\lambda _{\theta }$ of the matrix $Q(\alpha )$
is zero, then we have 
\[
y_{j_{\theta }j_{\mu }}^{\theta \mu }(\alpha )=y_{j_{\mu }j_{\theta }}^{\mu
\theta }(\alpha )=0:\forall \mu ,j\mu 
\]%
and the all remaining non-zero vectors $\{y_{j_{\upsilon }j_{\mu }}^{\nu \mu
}(\alpha )\}$ form the reduced basis of the ideal $U(\alpha )$ such that 
\[
y_{j_{\upsilon }j_{\mu }}^{\nu \mu }(\alpha )y_{j_{\sigma }j_{\rho
}}^{\sigma \rho }(\beta )=\delta ^{\alpha \beta }\delta ^{\mu \sigma
}\lambda _{\mu }y_{j_{\nu }j_{\rho }}^{\nu \rho }(\alpha ),\quad \mu ,\nu
,\rho ,\sigma \neq \theta, 
\]%
where $\mu ,\nu ,\rho ,\sigma $ label the irreducible representations of $%
S(n-1)$ that appears in the reduction of the induced representation $\Phi
^{\alpha }=\ind_{S(n-2)}^{S(n-1)}\varphi ^{\alpha }$, such that the
corresponding eigenvalues $\lambda _{\mu },$ $\lambda _{\nu },$ $\lambda
_{\rho },$ $\lambda _{\sigma }$ of the matrix $Q(\alpha )$ are non-zero. In
particular from this multiplication rule we get%
\[
(y_{j_{\nu }j_{\nu }}^{\nu \nu }(\alpha ))^{2}=\lambda _{\upsilon }y_{j_{\nu
}j_{\nu }}^{\nu \nu }(\alpha ) 
\]%
e.i $y_{j_{\upsilon }j_{\upsilon }}^{\nu \upsilon }(\alpha )$ is an
essential projector. Again from the last theorem in the Appendix~\ref{AppB} it
follows that the ideal $U(\alpha )$ is isomorphic with the matrix algebra $%
M(\rank Q(\alpha ),%
\mathbb{C}
)$ and the dimension of the ideal $U(\alpha )$ is 
\[
\dim U(\alpha )=(\rank Q(\alpha ))^{2}=((n-1)\dim \varphi ^{\alpha }-\dim \psi
^{\theta })^{2}. 
\]%
Note that in this case the vectors $\{u_{ij}^{ab}(\alpha )\}$, as linearly
dependent do not form a basis of $U(\alpha )$.

If all eigenvalues of the the matrix $Q(\alpha )$ are non zero (e.i when $%
\det Q(\alpha )\neq 0$), then from the Th.~\ref{th77} in the Appendix~\ref{AppB} we get that,
all the reduced vectors $\{$ $y_{j_{\nu }j_{\mu }}^{\nu \mu }(\alpha )\}$
defined in the (Eq.~\eqref{eq4}) are non-zero and they form a new basis of the algebra $%
U(\alpha )$, because in this case the vectors $\{u_{ij}^{ab}(\alpha )\}$
also form a basis of the algebra $\ U(\alpha )$ and the ideal $U(\alpha )$ is
isomorphic to the matrix algebra $M((n-1)\dim \varphi ^{\alpha },%
\mathbb{C}
)$. In both cases, when $\det Q(\alpha )\neq 0$ and $\det Q(\alpha )=0)$, we
may rescall the vectors of the the basis $\{y_{j_{\nu }j_{\mu }}^{\nu \mu
}(\alpha ):\mu ,\nu \neq \theta \},$ according the formula given in Th.~\ref{th77} in
the Appendix~\ref{AppB}, in the following way 
\[
y_{j_{\upsilon }j_{\mu }}^{\nu \mu }(\alpha )\rightarrow f_{j_{\upsilon
}j_{\mu }}^{\nu \mu }(\alpha )=\frac{1}{\sqrt{\lambda _{\mu }\lambda _{\nu }}%
}y_{j_{\upsilon }j_{\mu }}^{\nu \mu }(\alpha ) 
\]%
and a new basis $\{f_{j_{\upsilon }j_{\mu }}^{\nu \mu }(\alpha )\}$
satisfies the matrix multiplication rule%
\[
f_{j_{\upsilon }j_{\mu }}^{\nu \mu }(\alpha )f_{j_{\sigma }j_{\rho
}}^{\sigma \rho }(\beta )=\delta ^{\alpha \beta }\delta ^{\mu \sigma
}f_{j_{\nu }j_{\rho }}^{\nu \rho }(\alpha ),\quad \mu ,\nu ,\rho ,\sigma
\neq \theta 
\]%
and therefore we will call this basis reduced matrix basis.

From this matrix multiplication rule it follows that , for a fixed values of 
$\mu \neq \theta ,$ $j_{\mu }$, the basis vectors $\{f_{j_{\upsilon }j_{\mu
}}^{\nu \mu }(\alpha ):\nu \neq \theta ,j_{\nu }=1,..,\dim \psi ^{\nu }\}$
span the irreducible left modules of the algebra $A_{n}^{t_{n}}(d)$ (Th.~\ref{th5}
and Prop.~\ref{prop7}). Now we have to derive the transformation rule for the natural
generators $V(\sigma )^{t_n}\in A_{n}^{t_{n}}(d)$, $\sigma \in S(n)$ in
these irreducible representations. A laborious but purely technical
calculation shows that, they are the following

\begin{proposition}
\label{prop45}
Let $\{f_{j\nu j_{\mu }}^{\nu \mu }:\mu ,\nu ,\neq \theta \}$ be the reduced
matrix basis of the ideal $U(\alpha ),$ then for a fixed values of $\mu ,$ $%
j_{\mu }$ the vectors $\{$ $f_{j\nu j_{\mu }}^{\nu \mu }\}$ form a basis of
an irreducible left module $U_{j_{\mu }}^{\mu }(\alpha )=span_{%
\mathbb{C}
}\{f_{j\nu j_{\mu }}^{\nu \mu }:\nu \neq \theta ,j_{\nu }=1,..,\dim \psi
^{\nu }\}$ and we have 
\[
V(an)^{t_{n}}{}f_{j\nu j_{\mu }}^{\nu \mu }(\alpha )=\sum_{\rho \neq \theta
,j_{\rho }}(\sum_{k}\sqrt{\lambda _{\rho }}z^{+}(\alpha )_{j_{\rho }k}^{\rho
a}z(\alpha )_{kj_{\nu }}^{a\nu }\sqrt{\lambda _{\nu }}f_{j_{\rho }j_{\mu
}}^{\rho \mu }(\alpha ),
\]%
where the unitary matrix $Z(\alpha )=(z(\alpha )_{kj_{\nu }}^{a\nu })$ has
the form 
\[
 z(\alpha )_{kj_{\nu }}^{a\nu }=\frac{\dim \psi ^{\nu }}{\sqrt{%
N_{j_{\nu }}^{\nu }}(n-1)!}\sum_{\sigma \in S(n-1)}\psi _{j_{\nu }j_{\nu
}}^{\nu }(\sigma ^{-1})\delta _{a\sigma (q)}\varphi _{kr}^{\alpha
}[(a \ n-1)\sigma (q \ n-1)],
\]%
with 
\[
 N_{j_{\nu }}^{\nu }=\frac{\dim \psi ^{\nu }}{(n-1)!}%
\sum_{\sigma \in S(n-1)}\psi _{j_{\nu }j_{\nu }}^{\nu }(\sigma ^{-1})\delta
_{q\sigma (q)}\varphi _{rr}^{\alpha }[(q \ n-1)\sigma (q \ n-1)],
\]%
where the indices $q=1,..,n-1,$ $r=1,..,\dim \varphi
^{\alpha }$ are fixed and such that $N_{j_{\nu }}^{\nu }>0$
\ (see Th.59, 85 and Cor. 86) and for $\sigma _{n}\in S(n-1)$ we have 
\[
V(\sigma _{n})f_{j_{\nu }j_{\mu }}^{\nu \mu }(\alpha )=\sum_{j_{\nu ^{\prime
}}}\psi _{i_{\nu ^{\prime }}j_{\nu }}^{\nu }(\sigma _{n})f_{i_{\nu ^{\prime
}}j_{\mu }}^{\nu \mu }(\alpha ),
\]%
so for the elements $V(\sigma _{n})\in V(S(n-1)$ the irreducible left modul $%
U_{j_{\mu }}^{\mu }(\alpha )$ is a direct sum of the irreducible
representations $\psi ^{\nu }$ of the group $S(n-1),$ appearing in the
induced representation $\ind_{S(n-2)}^{S(n-1)}(\varphi ^{\alpha })$, such
that $\mu \neq \theta .$
\end{proposition}

\begin{remark}
\label{rem46}
If the matrix $Q(\alpha )$ is invertible then the above transformation rule
for the elements $V(\sigma _{n})^{t_{n}}=V(\sigma _{n})$ is the
transformation law for the induced representation $\ind_{S(n-2)}^{S(n-1)}(%
\varphi ^{\alpha })$ in the form reduced to the direct sum of irreducible
components ( Th.~\ref{th58} Eq.~\eqref{A1})
\end{remark}

\bigskip From transformation rules for the natural generators of the algebra 
$A_{n}^{t_{n}}(d)$ in the reduced matrix bases one can easily deduce the the
matrix forms of the irreducible representations of the algebra $%
A_{n}^{t_{n}}(d)$, which are given in Th.~\ref{th30} and Prop.~\ref{prop32} in Sec.~\ref{irreps1}

Unfortunately the transformation rules of the natural generators $V(\sigma
)^{t_{n}}\in A_{n}^{t_{n}}(d)$ in the reduced basis of the ideal $U(\alpha
), $ given in the Prop.~\ref{prop45} , are complicated because the matrix $Z(\alpha )$,
which appears in the transformation rules has complicated structure.
However, in case when $\det Q(\alpha )\neq 0,$ then transformation rules for
the irreducible left modules included in the ideal $U(\alpha )$ can be
written equivalently in simpler form given in the Prop.~\ref{prop42}, because in this
case, the vectors $\{u_{ij}^{ab}(\alpha ):a,b=1,...,n-1,\quad
i,j=1,...,w^{\alpha }\}$ also form a basis of the ideal $U(\alpha ).$The
transformation rules for the natural generators of the algebra $%
A_{n}^{t_{n}}(d)$ are much simpler in this basis but this is not a matrix
basis. Using the fact that $\det Q(\alpha )\neq 0,$ one can construct a new
matrix basis without using the matrix $Z(\alpha ),$ in which the
transformation rules remains as simple as in Prop.~\ref{prop42} The new basis is
defined as follows

\begin{definition}
\label{def47}
If the matrix $Q(\alpha )$ \ of the ideal $U(\alpha )$ is invertible then we
define the new elements of $U(\alpha )$ in the following way 
\[
e_{ij}^{ab}(\alpha )\equiv \sum_{qs}(Q^{-1})_{jq}^{bs}(\alpha
)u_{qi}^{sa}(\alpha ):a,b,s=1,...,n-1,\quad i,j,k=1,...,\dim \varphi
^{\alpha }, 
\]%
where $\alpha :\varphi ^{\alpha }\in V_{d}[S(n-2)].$
\end{definition}

\bigskip From the invertibility of the matrix $Q(\alpha )$ if follows that
the elements $\{e_{ij}^{ab}(\alpha )\}$ form a new basis of the ideal $M.$
For this basis of $U(\alpha )$ we have

\begin{proposition}
\label{prop48}
\[
e_{ij}^{ab}(\alpha )e_{kl}^{cd}(\beta )=\delta _{\alpha \beta }\delta
^{bc}\delta _{jk}e_{il}^{ad}(\alpha ), 
\]%
\[
V(\sigma _{ab})^{t_n }e_{ij}^{cq}(\alpha )=d^{\delta
_{ac}}\sum_{k=1}^{w^{\alpha }}\varphi _{ki}^{\alpha }[(b \ n-1)\widehat{\sigma }%
_{ab}(ac)(c \ n-1)]e_{kj}^{bq}(\alpha ),\quad a,b\neq n,
\]%
where $\widehat{\sigma }_{ab}=(bn)\sigma _{ab}=\sigma _{ab}(an)\in S(n-1)$
and $(b \ n-1)\widehat{\sigma }_{ab}(ad)(d \ n-1)\in S(n-2)$ and 
\[
V(\sigma _{n})e_{ij}^{cq}(\alpha )=\sum_{k=1}^{w^{\alpha }}\varphi
_{ki}^{\alpha }[(\sigma _{n}(c) \ n-1)\sigma _{n}(c \ n-1)]e_{kj}^{\sigma
_{n}(c) \ q}(\alpha ). 
\]
\end{proposition}

\bigskip So the transformation rule of the natural generators of the algebra 
$A_{n}^{t_{n}}(d)$ in the basis $\{e_{ij}^{ab}(\alpha )\}$ are expressed
only by matrices of the irreducible representations $\varphi ^{\alpha }$ of $%
S(n-2)$ and they are much simpler then analogous formulas for the reduced
matrix basis $\{f_{j\nu j_{\mu }}^{\nu \mu }\}.$

The multiplication rule of the basis elements $e_{ij}^{ab}(\alpha )$ is in
fact the matrix multiplication rule. To see this we introduce the following
notation%
\[
e_{ij}^{ab}(\alpha )=e_{AB}(\alpha ):A=(a-1)w^{\alpha }+i,\quad
B=(b-1)w^{\alpha }+j, 
\]%
where $w^{\alpha }=\dim \varphi ^{\alpha }$, then 
\[
e_{AB}(\alpha )e_{CD}(\beta )=\delta _{\alpha \beta }\delta _{BC}e_{AD}, 
\]%
and for fixed irreducible representation $\alpha $ the elements $%
\{e_{AB}(\alpha ):A,B=1,...,(n-1)w^{\alpha }\}$ span an ideal in $M$
isomorphic with matrix algebra $M((n-1)w^{\alpha },%
\mathbb{C}
).$ The elements $\{e_{ii}^{aa}(\alpha )=e_{AA}(\alpha ):\alpha
=1,....,k,\quad a=1,...,n-1,\quad i=1,...,\dim \varphi ^{\alpha }\}$ are
primitive (irreducible) idempotents of the algebra $A_{n}^{t_{n}}(d)$.

\bigskip From the above theorem easily follows

\begin{corollary}
\label{col49}
For any $x=1,...,n-1,\quad s=1,...,\dim \varphi ^{\alpha }$ the subspace 
\[
E_{s}^{x}(\alpha )\equiv \Span_{%
\mathbb{C}
}\{e_{js}^{ax}(\alpha ):a=1,...,n-1,\quad j=1,...,w^{\alpha }\}\subset
U(\alpha )\subset M 
\]%
is a left minimal ideal (i.e. is an irreducible representation) of the
algebra $A$ such that%
\[
U(\alpha )=\bigoplus _{x,s}E_{s}^{x}(\alpha )\simeq M((n-1)w^{\alpha },%
\mathbb{C}
) 
\]%
and all the irreducible representations $E_{s}^{x}(\alpha )$ for any $%
x=1,...,n-1,\quad s=1,...,\dim \varphi ^{\alpha }$ are isomorphic, so a
fixed irreducible representation $\varphi ^{\alpha }$ of $S(n-2)$ gives us $%
(n-1)w^{\alpha }$ isomorphic irreducible representations.
\end{corollary}

The fact that the subspace $E_{s}^{x}(\alpha )$ is the irreducible
representation space follows from the relation%
\[
E_{s}^{x}(\alpha )=\{ae_{ss}^{xx}(\alpha ):a\in A\}, 
\]%
where $e_{ss}^{xx}(\alpha )$ is a primitive (irreducible) idempotent which
follows immediately from the matrix multiplication rule of the elements $%
\{e_{js}^{ax}(\alpha )\}$ (Th.~\ref{th5}).

The matrix form of irreducible representations of the algebra $%
A_{n}^{t_{n}}(d)$ for the bases $\{e_{ij}^{ab}(\alpha )\}$ are the following

\begin{proposition}
\label{prop50}
When $\det Q(\alpha )\neq 0,$ the natural generators $V(\sigma
_{ab})^{t_n }$ and $V(\sigma _{n})$ of $A_{n}^{t_{n}}(d)$ are represented
in the irreducible representation space $E_{s}^{x}(\alpha )=\Span_{%
\mathbb{C}
}\{e_{js}^{ax}(\alpha ):a=1,...,n-1,\quad j=1,...,w^{\alpha }\}\subset
U(\alpha )$ by the matrices of the form 
\[
M_{e}^{\alpha }(V(\sigma )_{ab}^{t_n })_{cp,qi}=\delta _{bc}d^{\delta
_{aq}}\varphi _{pi}^{\alpha }[(c \ n-1)\widehat{\sigma }_{ab}(aq)(q \ n-1)],\qquad
\quad a,b\neq n 
\]%
and 
\[
M_{e}^{\alpha }(V(\sigma _{n}))_{pl,ai}=\delta _{p\sigma _{n}(a)}\varphi
_{li}^{\alpha }[(\sigma _{n}(a) \ n-1)\sigma _{n}(a \ n-1)] 
\]%
\ and we see that that these matrices does not depend on the indices $x,s$
so in each isomorphic irreducible representation space $E_{s}^{x}(\alpha )$ $%
x=1,...,n-1,\quad s=1,...,\dim \varphi ^{\alpha }$ the operators $V(\sigma
_{ab})^{t_n}$ are represented by the same matrices.
\end{proposition}

One can check that the matrices $M^{\alpha }(\sigma ^{t_n})$ that
represent the elements $V(\sigma )^{t_n}$ of the algebra $%
A_{n}^{t_{n}}(d)$ satisfy the multiplication rule of the algebra $%
A_{n}^{t_{n}}(d)$ (see Th.~\ref{th19} above) i.e. we have

\bigskip \quad 
\[
M^{\alpha }(\sigma _{ab}^{t_n})M^{\alpha }(\rho _{cq}^{t_n
})=d^{\delta _{aq}}M^{\alpha }((\sigma (d)n)\sigma _{ab}(dn)\rho
_{cq}^{t_n }), 
\]%
\[
M^{\alpha }(\sigma _{n})M^{\alpha }(\rho _{cq}^{t_n})=M^{\alpha }(\sigma
_{n}\rho _{cq}^{t_n}),\quad M^{\alpha }(\sigma _{n})M^{\alpha }(\rho
_{n})=M^{\alpha }(\sigma _{n}\rho _{n}), 
\]%
so we have to deal with the irreducible matrix representations of the
algebra $A_{n}^{t_{n}}(d).$

We may conclude this subsection formulating the following

\begin{corollary}
\label{col51}
The ideal $M$ is a direct sum of irreducible representations (left minimal
ideals), indexed by irreducible representations $\varphi ^{\alpha }\in
V_{d}[S(n-2)]$ of $S(n-2)$ and each irreducible representations has
multiplicity in $M$ equal to its dimension.
\end{corollary}

\subsection{Matrix ideals and left minimal ideals of the ideal $S$ of the
algebra $A_{n}^{t_{n}}(d).$}
\label{mat_ideals}

In this section we will describe the structure of ideal $S$ of the algebra $%
A_{n}^{t_{n}}(d),$ complementary to the ideal $M$ constructed in Prop.~\ref{prop27}.
From the general properties of the semi-simple algebras~\cite{Littlewood} it follows
that each matrix ideal $U(\alpha )$ in $M$, being semisimple, has a unit element $e_{\alpha }$%
, which is an idempotent of the algebra $A_{n}^{t_{n}}(d).$ These units of
the ideal $U(\alpha )$ are of the form 
\[
e_{\alpha }=\sum_{\nu \neq \theta ,j_{\nu }}f_{j_{\nu }j_{\nu }}^{\nu \nu }, 
\]%
in the reduced matrix basis, where $\theta :\lambda _{\theta }= 0$ and
when $\det Q(\alpha )\neq 0$ 
\[
e_{\alpha }=\sum_{a=1,...,n-1}\sum_{i=1,..,w^{\alpha }}e_{ii}^{aa}(\alpha ) 
\]%
in the basis $\{e_{ij}^{ab}(\alpha )\}$ of $U(\alpha ).$The sum of all these
units is a unit $e$ of the ideal $M$ and we have

\begin{proposition}
\label{prop52}
The element 
\[
e=\sum_{\alpha :\varphi ^{\alpha }\in V[S(n-2)]}e_{\alpha } 
\]%
is a unit of the ideal $M$ i.e. we have 
\[
\forall m\in M\quad me=em=m 
\]%
and 
\[
M=eA_{n}^{t_{n}}(d)e. 
\]
\end{proposition}

Using the properties of the idempotents $e,$ $\text{\noindent\(\mathds{1}\)}-e$ one gets the
decomposition of the algebra $A_{n}^{t_{n}}(d)$ given in the Proposition~\ref{prop27}.

\bigskip Using again the basics properties of the idempotents 
and the equations describing the linear dependence of the generating
elements $V(\sigma )^{t_{n}}:\sigma \in S(n)$ 
\[
h(\phi ^{\mu })>d\Rightarrow E_{ij}^{t_{n}}(\phi ^{\mu })=\frac{\dim \phi
^{\mu }}{n!}\sum_{\sigma \in S(n)}\phi _{ji}^{\mu }(\sigma
^{-1})V^{t_{n}}(\sigma )=0,~i,j=1,..,\dim \phi ^{\mu },
\]%
one can prove

\begin{theorem}
\label{th53}
The ideal is of the form 
\[
S=\Span_{%
\mathbb{C}
}\{V(\sigma _{n})(\text{\noindent\(\mathds{1}\)}-e):\sigma _{n}\in S(n-1)\}=A_{n-1}(d)(\text{\noindent\(\mathds{1}\)}%
-e), 
\]%
where $A_{n-1}(d)\subset A_{n}^{t_{n}}(d)$ the the elements $\{V(\sigma
_{n})(\text{\noindent\(\mathds{1}\)}-e):\sigma _{n}\in S(n-1)\}$ are natural generators of the
ideal $S.$ The ideal $S$ (as a left $A_{n}^{t_{n}}(d)-$module ) is a
representation of the algebra $A_{n}^{t_{n}}(d).$ The natural generators of
algebra $A_{n}^{t_{n}}(d)$ acts on the basis elements of $S$ in the
following way%
\[
\forall m\in M\quad m(V(\sigma _{n})(\text{\noindent\(\mathds{1}\)}-e))=0,\qquad 
\]%
\[
\forall V(\rho _{n})\in V_{d}[S(n-1)]\quad V_{d}(\rho _{n})(V_{d}(\sigma
_{n})(\text{\noindent\(\mathds{1}\)}-e))=V_{d}(\rho _{n}\sigma _{n})(\text{\noindent\(\mathds{1}\)}-e) 
\]%
i.e. the elements of the ideal $M$ acts on $S$ trivially as zero operators
(the elements of $M$ are not invertible). We have also 
\[
\forall V_{d}(\rho _{n}),V_{d}(\sigma _{n})\in V_{d}[S(n-1)]\quad V_{d}(\rho
_{n})(\text{\noindent\(\mathds{1}\)}-e)(V_{d}(\sigma _{n})(\text{\noindent\(\mathds{1}\)}-e))=V_{d}(\rho _{n}\sigma
_{n})(\text{\noindent\(\mathds{1}\)}-e) 
\]%
Thus the ideal $S$ is a representation of the algebra $%
\mathbb{C}
\lbrack (S(n-1)]$ i.e. is homomorphic with this algebra.
The elements $V(\sigma ):\sigma \in S(n)$ are, in general
linearly depended, but we have: if $d\geq n$ then 
\[
S=\bigoplus _{\nu }span_{%
\mathbb{C}
}\{E_{ij}^{S(n-1)}(\psi ^{\nu })(\text{\noindent\(\mathds{1}\)}-e)\}, 
\]%
where
\[
E_{ij}^{S(n-1)}(\psi ^{\nu })=\frac{\dim \psi ^{\nu }}{(n-1)!}\sum_{\sigma
_{n}\in S(n-1)}\psi _{ji}^{\nu }(\sigma _{n}^{-1})V(\sigma _{n}) 
\]%
 and $\nu $ runs over all irreducible representations of 
$S(n-1)$ and if $d<n$, then
\[
S=\bigoplus _{\nu }span_{%
\mathbb{C}
}\{E_{ij}^{S(n-1)}(\psi ^{\nu })(\text{\noindent\(\mathds{1}\)}-e)\}, 
\]%
where $\nu :h(\psi ^{\nu })<d$ and the elements $%
E_{ij}^{S(n-1)}(\psi ^{\nu })(1-e)$ in both cases are linearly independent.
\end{theorem}

From this it follows

\begin{corollary}
Let $\psi ^{\nu }$ $\nu =1,...,q$
be all irreducible representations of the group $S(n-1).$ Then 
\[
 S\simeq \bigoplus _{\nu }M(\dim \psi ^{\nu },%
\mathbb{C}
)
\]%
where, if $d\geq n$, the direct sum is over all irreducible
representations $\psi ^{\nu }$ of the group $S(n-1),$  and
when $d<n$  then the direct sum is over $\nu :h(\psi ^{\nu })<d.$%
 Each irreducible representation $\psi ^{\nu }$of the
group $S(n-1)$ defines an irreducible representation $\Psi ^{\nu
},$ of the algebra $A_{n}^{t_{n}}(d)$ in the following way%
\[
\Psi ^{\nu }(a)=%
\left\{ \begin{array}{ll}
0:a\in M, \\ 
\psi ^{\nu }(\sigma _{n}):a=\sigma _{n}\in S(n-1).%
\end{array} \right.
\]%
So in this representation the non-invertible elements of the ideal $M$ are
represented trivially by zero and therefore we call these representation of
the algebra $A_{n}^{t_{n}}(d)$ semi-trivial. The matrix forms of these
irreducible representations \ are simply matrix forms of the irreducible
representations of the group algebra $%
\mathbb{C}
\lbrack S(n-1)]\subset A_{n}^{t_{n}}(d)$ and zero matrices for the elements
of the ideal $M\subset A_{n}^{t_{n}}(d).$
\end{corollary}

\bigskip \bigskip \bigskip

Thus from the above results we get that the algebra $A_{n}^{t_{n}}(d)$
admits irreducible representations of two kinds. The irreducible
representations from the ideal $M$ are generated by the irreducible
representations of the group $S(n-2)$ and in these representations the
algebra $A_{n}^{t_{n}}(d)$ is represented nontrivialy. The second kind of
the irreducible representations of $A_{n}^{t_{n}}(d)$ are included in the
ideal $S$ and they are generated by irreducible representations of the group 
$S(n-1).$ These representations are in fact, the irreducible representations
of the subalgebra $%
\mathbb{C}
\lbrack V(S(n-1)]\subset $ $A_{n}^{t_{n}}(d)$ in which all the partially
transposed operators $V(\sigma _{ab})^{\prime }\in M$ are represented
trivially by zero operators.

\section{Acknowledgment}
Authors would like to thank Aram Harrow and Mary Beth Ruskai for valuable discussion. M.S. would like to thank also to Institute for Theoretical Physics, University of Wroc{\l}aw for hospitality, where some part of this work was done. M. S. is supported by the International PhD Project "Physics of future quantum-based information technologies": grant MPD/2009-3/4 from Foundation for Polish Science. Authors are also supported by NCN Ideas-Plus Grant(IdP2011000361). 
Part of this work was done in National Quantum Information Centre of Gda\'nsk.

\appendix
\section{The properties of the matrices $Q(\protect\alpha ).$}
\label{AppA}

In this appendix we derive the basic properties of the matrices $Q(\alpha )$
were defined in Def~\ref{def28}.

\bigskip Directly from the definition it follows that the matrix $Q(\alpha ) 
$ are hermitian. In fact we have

\begin{proposition}
\label{prop55}
For any unitary representation $\varphi ^{\alpha }$ of the group $S(n-2)$
the matrix $Q(\alpha )$ is hermitian. This result holds even when the
representation $\varphi ^{\alpha }$ is not reducible.
\end{proposition}

\bigskip The structure of the matrix $Q(\alpha )$ is strictly connected with
the representation of the group $S(n-1)$ induced by the irreducible
representation $\varphi ^{\alpha }$ of its subgroup $S(n-2),$ namely we have

\begin{proposition}
\label{prop56}
The matrix $Q(\alpha )$ has the following form%
\[
Q(\alpha )\equiv Q_{n-1}(\alpha )=\sum_{1\leq a<b\leq n-1}\Phi ^{\alpha
}(ab)+(d-\digamma )\Phi ^{\alpha }(\id), 
\]%
where $\digamma =\frac{(n-2)(n-3)}{2}\frac{\chi ^{\alpha }(12)}{\dim \varphi
^{\alpha }}$, $\chi ^{\alpha }$ is the character of the representation $%
\varphi ^{\alpha }$ and $\Phi ^{\alpha }=\ind_{S(n-2)}^{S(n-1)}(\varphi
^{\alpha }).$ Thus we see that the matrix $Q_{n-1}(\alpha )$ is the sum of
matrices representing, in the induced representation of $S(n-1),$ all
transposition in the group $S(n-1)$ plus the unit matrix multiplied by the
number $d-F$ depending only on $\ d$ and the parameters of the
representation ($\varphi ^{\alpha },V^{\alpha })$ of the group $S(n-2).$
\end{proposition}

\bigskip In the proof of this Proposition, as well in the of the main
theorem of this Appendix below, we will need the following

\begin{theorem}
\label{th57}
Let ($\varphi ^{\alpha },V^{\alpha })$ be an irreducible representation of
the group $S(m)$ with the character $\chi ^{\alpha }$. Then%
\[
\sum_{\sigma \in K}\varphi ^{\alpha }(\sigma )=\frac{n_{K}}{\dim \varphi
^{\alpha }}\chi ^{\alpha }(\sigma ^{t_n })\id_{V^{\alpha }}, 
\]%
where $K$ is a class of conjugated elements, $n_{K}$ is the number of
permutations in the class $K$ and $\sigma ^{t_n}$ is any representant of
the class $K.$ In particular for the class of the transpositions we have%
\[
\sum_{(ab)\in S(m)}\varphi ^{\alpha }(ab)=\frac{m(m-1)}{2}\frac{\chi
^{\alpha }(12)}{\dim \varphi ^{\alpha }}\id_{V^{\alpha }}. 
\]
\end{theorem}

The statement of the theorem follows from the fact that $\sum_{\sigma \in
K}\varphi ^{\alpha }(\sigma )$ belongs to the center of the algebra $\varphi
^{\alpha }(S(m))$, so from the irreducibility of the representation ($%
\varphi ^{\alpha },V^{\alpha })$ and Schur Lemma we get the result. Now we
go back to the proof of the Proposition~\ref{prop56}.

\begin{proof}
Let ($\varphi ^{\alpha },V^{\alpha })$ be an irreducible representation of
the group $S(n-2)$ characterized by the partition(or the Young diagram) $%
\alpha .$ Consider the representation $\Phi ^{\alpha }$ of the group $S(n-1)$
induced by the representation ($\varphi ^{\alpha },V^{\alpha })$ of the
subgroup $S(n-2)\subset S(n-1).$ The block matrices $\Phi ^{\alpha }(\sigma
),$ $\sigma \in S(n-1)$ have dimension $N=(n-1)\dim \varphi ^{\alpha }$. If
we chose the representatives of the cosets $S(n-1)\backslash S(n-2)$ in the
standard way i.e. as $(1~n-1),$ $(2~n-1),..,(n-2~n-1)$ and use the standard
formula for induced matrix representation, then we have 
\[
\Phi ^{\alpha }(ab)=\left( 
\begin{array}{cccccccc}
\varphi ^{\alpha }(ab) & 0 & . &  &  &  & . & 0 \\ 
0 & . &  &  &  &  &  & . \\ 
&  & 0 &  &  & \varphi ^{\alpha }(ab) &  &  \\ 
&  &  & . &  &  &  &  \\ 
&  &  &  & \varphi ^{\alpha }(ab) &  &  &  \\ 
&  & \varphi ^{\alpha }(ab) &  &  & 0 &  &  \\ 
&  &  &  &  &  & . & 0 \\ 
0 &  &  &  &  &  &  & \varphi ^{\alpha }(ab)%
\end{array}%
\right) , 
\]%
for $a,b\leq n-2$ and%
\[
\Phi ^{\alpha }(n-1b)=\left( 
\begin{array}{cccccccc}
\varphi ^{\alpha }(1b) & 0 &  &  &  &  & 0 & 0 \\ 
0 & \varphi ^{\alpha }(2b) &  &  &  &  &  & 0 \\ 
&  & . &  &  &  &  &  \\ 
&  &  & \varphi ^{\alpha }(b-1b) &  &  &  &  \\ 
&  &  &  & 0 &  &  & 1 \\ 
&  &  &  &  & . &  &  \\ 
0 &  &  &  &  &  & \varphi ^{\alpha }(n-2b) & 0 \\ 
0 & 0 &  &  & 1 &  & 0 & 0%
\end{array}%
\right) , 
\]%
where $b\leq n-2.$ Comparing these matrices with the matrix $Q(\alpha
)\equiv Q_{n-1}(\alpha )$ given in its definition we get the statement of
the Proposition 
\[
Q(\alpha )\equiv Q_{n-1}(\alpha )=\sum_{1\leq a<b\leq n-1}\Phi ^{\alpha
}(ab)+(d-\digamma )\Phi ^{\alpha }(\id), 
\]%
where $\digamma =\frac{(n-2)(n-3)}{2}\frac{\chi ^{\alpha }(12)}{\dim \varphi
^{\alpha }}$ and we see that the matrix $Q_{n-1}(\alpha )$ belongs to the
algebra $\Phi ^{\alpha }[S(n-1)].$
\end{proof}

The solution of the most important problem of eigenvalues and eigenvectors
of the matrix $Q(\alpha )$ is given in the following main theorem of this
Appendix

\begin{theorem}
\label{th58}
\begin{enumerate}[a)]
\item Let $\varphi ^{\alpha }$ any irreducible representation of the group $%
S(n-2),$ $\alpha =(\alpha _{1},...,\alpha _{k})$ its partition and $%
Y^{\alpha }$ the corresponding Young diagram . The rank $r=r(\alpha )$ (or $%
r( $ $Y^{\alpha })$ ) of the partition $\alpha $ is the length of diagonal
of its Young diagram. Suppose that for some $i$ the sequence $\nu =(\alpha
_{1},..\alpha _{i}+1,..,\alpha _{k})$ is a partition of $n-1,$ so it defines
an irreducible representation $\psi ^{\nu }$ of the group $S(n-1)$ and for
the Young diagrams it means that, the Young diagram $Y^{\nu }$ \ is obtained
from the Young diagram $Y^{\alpha }$ by adding, in the $i-$th row , one box.
Then the corresponding matrix $Q_{n-1}(\alpha )$ has the following eigenvalue
\begin{enumerate}[i)]
\item if $r(Y^{\alpha })=r(Y^{\nu }),$ then 
\[
\qquad \lambda _{\nu }=d+\alpha _{i}+1-i,\quad \quad i=1,...,k+1, 
\]%
and if $i=k+1$ we set $\alpha _{k+1}=0.$

\item if $r(Y^{\alpha })+1=r(Y^{\nu })$ which may occur only if $i=r+1,$ then 
\[
\qquad \lambda _{\nu }=d.\quad 
\]%
The case $ii)$ describes the situation when adding, in a proper way one box
to Young diagram $Y^{\alpha }$ we extend its diagonal. The multiplicity of
the eigenvalue $\lambda _{\nu }$ is equal to $dim\psi ^{\nu }$ and the
number of pairwise distinct eigenvalues of the matrix $Q_{n-1}(\alpha )$ is
equal to Young diagrams $Y^{\nu }$ that one can obtain from the Young
diagram $Y^{\alpha }$ by adding, in a proper way, one box.
\end{enumerate}
\item  The unitary matrix $Z(\alpha )=(z(\alpha )_{kj_{\nu }}^{a\nu })$
which reduce the induced representation $\Phi ^{\alpha
}=\ind_{S(n-2)}^{S(n-1)}\varphi ^{\alpha }$ into the irreducible components
has the form 
\[
\Rrightarrow z(\alpha )_{kj_{\nu }}^{a\nu }=\frac{1}{\sqrt{N_{j_{\nu}}}^{\nu}}(E_{j_{\nu }j_{\nu }}^{\nu
})_{aq}^{kr}=\frac{\dim \psi ^{\nu }}{\sqrt{%
N_{j_{\nu }}^{\nu }}(n-1)!}\sum_{\sigma \in S(n-1)}\psi _{j_{\nu }j_{\nu
}}^{\nu }(\sigma ^{-1})\delta _{a\sigma (q)}\varphi _{kr}^{\alpha
}[(a \ n-1)\sigma (q \ n-1)],\Lleftarrow 
\]%
with 
\[
 N_{j_{\nu }}^{\nu }\equiv (E_{j_{\nu }j_{\nu }}^{\nu
})_{rr}^{qq}=\frac{\dim \psi ^{\nu }}{(n-1)!}\sum_{\sigma \in S(n-1)}\psi
_{j_{\nu }j_{\nu }}^{\nu }(\sigma ^{-1})\delta _{q\sigma (q)}\varphi
_{rr}^{\alpha }[(q \ n-1)\sigma (q \ n-1)],
\]%
where $\psi ^{\nu }$ are representations of the group $S(n-1)$ whose Young
diagrams are obtained from the Young diagram $\alpha $ of $\varphi ^{\alpha
} $ by adding, in a proper way, one box and $(\psi _{j_{\nu }j_{\nu }}^{\nu
}(\sigma ))$ is a matrix form of $\sigma \in S(n-1)$ in the representation $%
\psi ^{\nu }$, $E_{j_{\nu }j_{\nu }}^{\nu }$ is a
hermitian projector of rank one in the representation space $\Phi ^{\alpha
} $ defined by $\psi ^{\nu }$ (see Def. 79 in App. C) and
the double index $(q,r)$ is fixed and chosen in such a way that $%
N_{j_{\nu }}^{\nu }>0$, which according Th. 84 and Cor.85, is always
possible. The double index $(q,r)$ shows which non-zero column from the projector $\left(E_{j_{\nu}j_{\nu}}^{\nu}\right)_{kr}^{aq}$ we chose to construct the matrix $z(\alpha)^{a\nu}_{kj_{\nu}}$. We have also 
\[
\label{A1}
\sum_{ak}\sum_{bl}z^{\dagger}(\alpha )_{j_{\rho }k}^{\rho a}\Phi ^{\alpha }(\sigma
)_{kl}^{ab}z(\alpha )_{lj_{\mu }}^{b\mu }=\delta ^{\rho \mu }\psi _{j_{\rho
}j_{\mu }}^{\mu }(\sigma ).
\]%
In particular 
\be
\label{A2}
\sum_{ak}\sum_{bl}z^{\dagger}(\alpha )_{j_{\rho }k}^{\rho a}Q(\alpha
)_{kl}^{ab}z(\alpha )_{lj_{\mu }}^{b\mu }=\delta ^{\rho \mu }\delta
_{j_{\rho }j_{\mu }}\lambda _{\mu }
\ee
and the columns of the matrix $Z(\alpha )=(z(\alpha )_{kj_{\nu }}^{a\nu })$
are eigenvectors of the matrix $Q(\alpha ).$
\end{enumerate}
\end{theorem}

\begin{remark}
\label{rem60}
The part a) of this theorem gives an explicit and remarkable simple
dependence or eigenvalues of the matrix $Q_{n-1}(\alpha )$ on the partition $%
\alpha =(\alpha _{1},...,\alpha _{k})$ which defines the the irreducible
representation $\varphi ^{\alpha }$ and consequently the matrix $%
Q_{n-1}(\alpha )$.
\end{remark}

The proof of this theorem uses for lemmas which are succeeding steps of this
proof. The first is the the following

\begin{lemma}
\label{lem61}
Let $\varphi ^{\alpha }$ any irreducible representation of the group $%
S(n-2), $ $\alpha $ its partition (or Young diagram), $\chi ^{\varphi
^{\alpha }}$ its character and let $\psi ^{\nu }$ be all irreducible
representations of the group $S(n-1)$ whose Young diagrams are obtained
from the Young diagram $\alpha $ of $\varphi ^{\alpha }$ by adding, in a
proper way, one box. By $\chi ^{\psi ^{\nu }}$ we denote their characters,
where $\nu $ is the partition of $n-1$ which labels the representation $\psi
^{\nu }$. Then the distinct eigenvalues of the matrix $Q(\alpha )$ generated
by the irreducible representation $\varphi ^{\alpha }$ of $S(n-2)$ are
labelled by the partitions $\nu $ and are of the form%
\[
\lambda _{\nu }=d+\frac{(n-1)(n-2)}{2}\frac{\chi ^{\psi ^{\nu }}(ab)}{\dim
\psi ^{\nu }}-\frac{(n-2)(n-3)}{2}\frac{\chi ^{\varphi ^{\alpha }}(ab)}{\dim
\varphi ^{\alpha }}, 
\]%
where $(ab),$ $a,b<n-2$ is an arbitrary transposition in $S(n-2),$ the
eigenvalue $\lambda _{\nu }$ has multiplicity $\dim \psi ^{\nu }.$
\end{lemma}

\begin{proof}
The induced representation $\Phi ^{\alpha }$ of the group $S(n-1)$ is
reducible, in fact we have~\cite{Fulton}
\[
\Phi ^{\alpha }=\bigoplus _{\lambda }\psi ^{\lambda }, 
\]%
where $\psi ^{\lambda }$ are irreducible representations of the group $%
S(n-1) $ whose Young diagrams are obtained from the Young diagram $\alpha $
of $\varphi ^{\alpha }$ by adding, in a proper way, one box. This means
that one can transform the matrix $Q_{n-1}(\alpha )$ to the reduced block
diagonal $Q_{n-1}^{R}(\alpha )$ form by a similarity transformation
generated by a block matrix $\ Z=(z_{ij_{\lambda }}^{a\lambda }),$ $%
j_{\lambda }=1,..,\dim \psi ^{\lambda }$ such that 
\[
Q_{n-1}^{R}(\alpha )=Z^{-1}Q_{n-1}(\alpha )Z=\bigoplus _{\lambda }(\sum_{1\leq
a<b\leq n-1}\psi ^{\lambda }(ab))+(d-\digamma )\text{\noindent\(\mathds{1}\)}, 
\]%
where the representations $\psi ^{\lambda }$ of $S(n-1)$ are irreducible so
from the Theorem we get 
\[
\sum_{1\leq a<b\leq n-1}\psi ^{\lambda }(ab)=\frac{(n-1)((n-2)}{2}\frac{\chi
^{\psi ^{\lambda }}(12)}{\dim \psi ^{\lambda }}\text{\noindent\(\mathds{1}\)}_{\psi ^{\lambda }} 
\]%
and it is clear that the similarity transformation simply diagonalizes the
matrix $Q_{n-1}(\alpha )$ giving all eigenvalues with their multiplicities.
The matrix $Q(\alpha )$, as a block matrix of an induced representation, has
in natural way, double indices $a,j:a=1,...,n-1,\quad j=1,...,\dim \varphi
^{\alpha }$, where $\varphi ^{\alpha }$ is an irreducible representation of
the group $S(n-2).$ After the reduction of the induced representation $\Phi
^{\alpha }$ of the group $S(n-1)$ to the form%
\[
\Phi ^{\alpha }=\bigoplus _{\lambda }\psi ^{\lambda }, 
\]%
where $\psi ^{\lambda }$ are irreducible representations of the group $%
S(n-1),$ the matrix $Q(\alpha )$ is transformed to the reduced form $%
Q_{n-1}^{R}(\alpha )$ which has indices naturally related to its block
diagonal structure%
\[
Q_{n-1}^{R}(\alpha )=((Q_{n-1}^{R})_{j\nu j_{\mu }}^{\nu \mu }(\alpha )) 
\]%
and 
\[
(Q_{n-1}^{R})_{j\nu j_{\mu }}^{\nu \mu }(\alpha
)=\sum_{ai,bk}(z^{-1})_{j_{\nu }i}^{\nu a}Q_{jk}^{ab}(\alpha )z_{kj_{\mu
}}^{b\mu } 
\]
\end{proof}

The formula for the eigenvalue of the matrix $Q_{n-1}(\alpha )$, given in
the Lemma~\ref{lem61}, is not entirely analytical because there is no analytical
formula for irreducible characters of the group $S(n).$ For arbitrary $n\in 
\mathbb{N}
.$ However for a given $n-2\in 
\mathbb{N}
$ and given partition $\alpha $ of $n-2$ and can always calculate the value
of the corresponding character on the class of transpositions as well the
value of the irreducible character $\chi ^{\psi ^{\nu }}$ on the class of
transpositions.

The formula for the eigenvalues of the matrix $Q_{n-1}(\alpha )$ derived in
the Lemma~\ref{lem61} express the eigenvalues by characters if irreducible
representations of $S(n-2)$ and $S(n-1).$ The next step necessary for the
proof of the Th.~\ref{th58}, is to derive the formula for the eigenvalues of the
matrix $Q_{n-1}(\alpha )$ which allows to write down (practically without
calculations) these eigenvalues using some characteristic of Young diagram
of the irreducible representation $\varphi ^{\alpha }$ of the group $S(n-2).$
For this we will need the following

\begin{definition}~\cite{Fulton}
\label{def62}
Let $\alpha =(\alpha _{1},...,\alpha _{k})$ be a partition of $n-2$ and \ $%
Y^{\alpha }$ corresponding Young diagram. We define the rank $r=r(\alpha )$
(or $r($ $Y^{\alpha })$ ) of the partition $\alpha $ is the length of
diagonal of its Young digram. Let $a_{i}$, $b_{i}$, $i=1,...,r$ be
respectively, the number of the boxes below and to the right of the $i-$th
box in the diagonal, reading from the upper left to the lower right. So we
have 
\[
a_{1}>a_{2}>..>a_{r},\quad b_{1}>b_{2}>..>b_{r}, 
\]%
and 
\[
\left( 
\begin{array}{cccccc}
a_{1} & a_{2} & . & . & . & a_{r} \\ 
b_{1} & b_{2} & . & . & . & b_{r}%
\end{array}%
\right) 
\]%
is called the characteristic of the partition $\alpha $ or equivalently of
its Young diagram.
\end{definition}

\begin{example}
\label{ex63}
$\alpha =(4,2,2)$ so 
\[
Y^{\alpha }=%
\begin{array}{cccc}
\bullet & \bullet & \bullet & \bullet \\ 
\bullet & \bullet &  &  \\ 
\bullet & \bullet &  & 
\end{array}%
\]%
then the rank is $r=2$ and the characteristic is 
\[
\left( 
\begin{array}{cc}
2 & 1 \\ 
3 & 0%
\end{array}%
\right) 
\]
\end{example}

The important point in the proof of the part a) of the main theorem is the
following result proved by Frobenius~\cite{Fulton}.

\begin{theorem}
\label{th64}
Let $\varphi ^{\alpha }$ the irreducible representation of the group $S(m)$
indexed by the partition $\alpha $ or rank $r$ and characteristic 
\[
\left( 
\begin{array}{cccccc}
a_{1} & a_{2} & . & . & . & a_{r} \\ 
b_{1} & b_{2} & . & . & . & b_{r}%
\end{array}%
\right) 
\]%
then the value of its character on the class of transpositions is give by
very simple formula 
\[
\chi ^{\alpha }(12)=\frac{\dim V^{\alpha }}{m(m-1)}%
\sum_{i=1}^{r}(b_{i}(b_{i}+1)-a_{i}(a_{i}+1)). 
\]
\end{theorem}

A substitution of this formula into the formula for eigenvalues of the
matrix $Q_{n-1}(\alpha )$ derived in the Lemma~\ref{lem61} gives us

\begin{lemma}
\label{lem65}
Let $\varphi ^{\alpha }$ any irreducible representation of the group $%
S(n-2), $ $\alpha =(\alpha _{1},...,\alpha _{k})$ its partition of rank $%
r=r(\alpha ) $ with characteristic 
\[
\left( 
\begin{array}{cccccc}
a_{1} & a_{2} & . & . & . & a_{r} \\ 
b_{1} & b_{2} & . & . & . & b_{r}%
\end{array}%
\right) 
\]%
and the Young diagram $Y^{\alpha }$. Suppose that for some $i$ the sequence $%
\nu =(\alpha _{1},..\alpha _{i}+1.,\alpha _{k})$ is a partition of $n-1,$ so
it defines an irreducible representation $\psi ^{\nu }$ of the group $S(n-1)$
and for the Young diagrams it means that the Young diagram $Y^{\nu }$ \ is
obtained from the Young diagram $Y^{\alpha }$ by adding, in the $i-$th row ,
one box. Then the corresponding matrix $Q_{n-1}(\alpha )$ has the following
eigenvalue%
\begin{enumerate}[a)]
\item 
\[
\lambda _{\nu }=d+b_{i}+1\quad if\quad i\leq r, 
\]%
\item
\[
\lambda _{\nu }=d-(a_{j}+1):j=\lambda _{i}+1\quad if\quad i>r\quad
and\quad r(Y^{\alpha })=r(Y^{\nu }), 
\]%
\item
\[
\lambda _{\nu }=d\quad if\quad i=r+1\quad and\quad r(Y^{\alpha
})+1=r(Y^{\nu }). 
\]%
\end{enumerate}
The case $c)$ describes the situation when adding, in a proper way one box
to Young diagram $Y^{\alpha }$ we extend its diagonal. The multiplicity of
the eigenvalue $\lambda _{\nu }$ is equal to $dim\psi ^{\nu }$ and the
number of pairwise distinct eigenvalues of the matrix $Q_{n-1}(\alpha )$ is
equal to Young diagrams $Y^{\nu }$ that one can obtain from the Young
diagram $Y^{\alpha }$ by adding, in a proper way, one box.
\end{lemma}

Thus derivation of the eigenvalues of the matrix $Q_{n-1}(\alpha )$ using
this Lemma need not any calculations. It is enough to look what Young
diagrams $Y^{\nu }$ one can obtain from the Young diagram $Y^{\alpha }$ by
adding, in a proper way, one box and then the place of this extra box in $%
Y^{\nu }$ determines what number from the characteristic of the Young
diagram $Y^{\alpha }$ one should add to $d\pm 1$ or to $d$ in order to get
an eigenvalue of the matrix $Q_{n-1}(\alpha ).$

The last step of the proof of the part a) of the Th.~\ref{th58} is a simple analysis
of dependence between the numbers of the partition $\alpha =(\alpha
_{1},...,\alpha _{k})$ of the irreducible representation $\varphi ^{\alpha }$
on the numbers of its characteristic $\left( 
\begin{array}{cccccc}
a_{1} & a_{2} & . & . & . & a_{r} \\ 
b_{1} & b_{2} & . & . & . & b_{r}%
\end{array}%
\right) ,$ which leads to the remarkable simply formula for the eigenvalues
of the matrix $Q_{n-1}(\alpha )$ given in the part $a)$ of the main theorem
in this Appendix.

The proof of the part $b)$ of the main theorem consists of a rather
laborious but purely technical calculation where, in order to show the
statements of the the theorem, one uses the orthogonality relations for
irreducible representations of the symmetric group which are of the form%
\[
\frac{1}{m!}\sum_{\sigma \in S(m)}\varphi _{ij}^{\alpha }(\sigma
^{-1})\varphi _{kl}^{\beta }(\sigma )=\frac{1}{\dim \varphi ^{\alpha }}%
\delta ^{\alpha \beta }\delta _{il}\delta _{jk}, 
\]%
where $\varphi ^{\alpha }$ and $\varphi _{kl}^{\beta }$ are irreducible
representations of the symmetric group $S(m).$

The explicit formulas for the eigenvalues and eigenvectors of the matrix $%
Q_{n-1}(\alpha )$ \ derived in the Th.~\ref{th58} leads to several conclusions.

For any Young diagram $Y^{\alpha }$ we may always add one box at the end of
first line as well below the last line of $Y^{\alpha },$ therefore from the
main theorem one deduce

\begin{corollary}
\label{col66}
Under assumptions of the theorem~\ref{th58} we have: for $\alpha =(\alpha
_{1},...,\alpha _{k})$ the numbers 
\[
\lambda _{\nu }=d+\alpha _{1},\qquad \lambda _{\nu }=d-k 
\]%
are always the eigenvalues of the matrix $Q_{n-1}(\alpha )$ and 
\[
d=k\Rightarrow \lambda _{k}=0. 
\]
\end{corollary}

The existence of zero eigenvalues of the matrices $Q_{n-1}(\alpha )$ is
essential for the inversibility of these matrices. The following conclusion
describes when the zero eigenvalues in matrices $Q_{n-1}(\alpha )$ do
appear.

\begin{corollary}
\label{col67}
The matrix $Q_{n-1}(\alpha )$ has an eigenvalue equal to zero iff \ for some 
$i$ the sequence $\nu =(\alpha _{1},..\alpha _{i}+1,..,\alpha _{k})$ is a
partition of $n-1,$ so it defines an irreducible representation $\psi ^{\nu
} $ of the group $S(n-1)$ and 
\[
d=i-\alpha _{i}-1. 
\]%
The multiplicity of this zero eigenvalue is equal to $\dim \psi ^{\nu }$ and
such a zero eigenvalue may appear only for one irreducible representation $%
\psi ^{\nu }.$
\end{corollary}

Directly from the formula for the eigenvalues of the matrix $Q_{n-1}(\alpha
),$ given in Th.59 one gets also

\begin{corollary}
\label{68}
If $d>n-2$ then the matrix $Q(\alpha )$ is (strictly) positive i.e. $%
Q(\alpha )>0$ and consequently if $d>n-2$ the matrix $Q(\alpha )$ is
invertible for any irreducible representation $\varphi ^{\alpha }.$
\end{corollary}

\bigskip An alternative proof of this statement is given in~\cite{Studzinski1}.

Using the formula of this theorem one can easily calculate the eigenvalues
for particular matrices $Q_{n-1}(\alpha )$ .

The simplest cases are the following

\begin{example}
\label{ex69}
When the irreducible representation $\varphi ^{\alpha }$ of the group $%
S(n-2) $ is the identity representation then from the theorem we get 
\[
\lambda _{0}=d+n-2,\qquad \lambda _{1}=d-1, 
\]%
with the multiplicities $1$ and $n-2$ respectively.
\end{example}

Similarly

\begin{example}
\label{ex70}
When the irreducible representation $\varphi ^{\alpha }$ of the group $
S(n-2)$ is the $\sgn$ representation then from the theorem we get
\end{example}

\[
\lambda _{0}=d-(n-2),\qquad \lambda _{1}=d+1, 
\]%
with the multiplicities $1$ and $n-2$ respectively.

These results one can obtain independently, by a direct calculation of
eigenvalues of the matrices $Q(\alpha )$ using standard procedures, which
are easily applicable in these simplest cases. The next examples are less
simple

\begin{example}
\label{ex71}
Let $\alpha =(2,1,...,1)$ be a partition of $n-2$. So $r(\alpha )=1$. Now by
adding to the Young diagram $Y^{\alpha }$ one box one can obtain three young
diagrams $Y^{\nu }$ describing the irreducible representations of the group $%
S(n-1).$

\begin{enumerate}[I)]
\item $\nu =(3,1,1,...,1)$ \ and this the case $a)$ $i)$ in the main theorem,
so the eigenvalue takes the form%
\[
\lambda _{\nu }=d+2 
\]%
\item $\nu =(2,2,1,...,1)$ so and this is the case $a)$ $ii)$, where the
diagonal of the Young diagram $Y^{\alpha }$ is extended therefore 
\[
\lambda _{\nu }=d 
\]%
\item $\nu =(2,1,1,...,1)$ then we have to deal again with the case $a)$ $i)$
of the theorem so we get 
\[
\lambda _{\nu }=d-n+3. 
\]
\end{enumerate}
\end{example}

\begin{example}
\label{ex72}
Let $\alpha =(n-3,1,...,0)$ be a partition of $n-2$. So $r(\alpha )=1$ and
by adding to the Young diagram $Y^{\alpha }$ one box one obtain three Young
diagrams $Y^{\nu }$ describing the irreducible representations of the group $%
S(n-1).$

\begin{enumerate}[I)]
\item $\nu =(n-2,1,0,...,0)$ \ and this the case $a)$ $i)$ in the main theorem,
so the eigenvalue takes the form%
\[
\lambda _{\nu }=d+n-3 
\]%
\item $\nu =(n-3,2,0,...,0)$ and this is the case $a)$ $ii)$, where the
diagonal of the Young diagram $Y^{\alpha }$ is extended therefore%
\[
\lambda _{\nu }=d 
\]%
\item $\nu =(n-3,1,1,...,0)$ now it is the case $a)$ $i)$ in the Theorem so
we get 
\[
\lambda _{\nu }=d-2. 
\]
\end{enumerate}
\end{example}

\section{}
\label{AppB}

In this Appendix we will consider the properties of the following algebra

\begin{definition}
\label{def73}
Let $A=(a_{ij})\in M(m,%
\mathbb{C}
)$, then the algebra $X_{A}$ is defined as 
\[
X_{A}=\Span_{%
\mathbb{C}
}\{x_{ij}:i,j=1,...,m\}, 
\]%
where 
\[
x_{ij}x_{kl}=a_{jk}x_{il}. 
\]%
and we do not assume that the elements $\{x_{ij}\}$ are linearly
independent, thus the algebra $X_{A}$ ic a complex finite-dimensional
algebra of the dimension at most $m^{2}.$
\end{definition}

Obviously the properties of the algebra $X_{A}$ depends on the properties of
the matrix $A.$ In fact we have

\begin{theorem}
\label{th74}
Suppose that the matrix $A$ \ in the algebra $X_{A}$ is invertible then we
have two possibilities

a) $X_{A}=\{0\}$ e.i the algebra $X_{A}$ a zero algebra,

b) If $X_{A}\neq \{0\}$ then algebra $X_{A}$ is isomorphic to the matrix
algebra $M(m,%
\mathbb{C}
)$ and the elements $\{x_{ij}:i,j=1,...,m\}$ are linearly independent, in
particular $x_{ij}\neq 0:i,j=1,...,m$, and form the basis of $X_{A}.$ and in
this case the unit of the algebra $X_{A}$ is of the form%
\[
\text{\noindent\(\mathds{1}\)}=\sum_{i,j=1,..,m}(a_{ij}^{-1})x_{ij}, 
\]%
where $A^{-1}=(a_{ij}^{-1}).$
\end{theorem}

\begin{proof}
\bigskip Suppose that $x_{ij}=0$ for some indices $i,j=1,...,m$. Then from
multiplication law for the algebra $X_{A}$ we get%
\[
\forall k,l=1,...,m,\quad x_{ij}x_{kl}=a_{jk}x_{il}=0 
\]
and because the matrix $A=(a_{ij})$ is invertible (so it has no zero columns
or zero rows) then we get 
\[
x_{ij}=0\Rightarrow \forall k,l=1,...,m\quad x_{ik}=x_{li}=0 
\]%
and consequently $\forall k,l=1,...,m\quad x_{kl}=0.$ If $X_{A}\neq \{0\}$
then defining new basis 
\[
y_{kj}=\sum_{i=1,..,m}(a_{ik}^{-1})x_{ij} 
\]
we get 
\[
y_{ij}y_{kl}=\delta _{jk}y_{il}. 
\]
\end{proof}

So we see that the that the invertibility of the matrix $A$ strongly
determines the properties of the vectors $\{x_{ij}:i,j=1,...,m\}$ which span
the algebra $X_{A}.$ We have also

\begin{theorem}
\label{th75}
Suppose the the vectors $\{x_{ij}:i,j=1,...,m\}$ which span the algebra $%
X_{A}$ are linearly independent and $\det (A)=0$ then there exist in the
algebra $X_{A}$ a nonzero properly nilpotent element and consequently the
algebra $X_{A}$ is not semisimple.
\end{theorem}

\begin{proof}
From the assumption \bigskip $\det (A)=0$ it follows that there exist s
nonzero vector $u=(u_{1},u_{2},....,u_{m})\in 
\mathbb{C}
^{m}$ such that 
\[
Au=0\Leftrightarrow \sum_{k=1,..,m}(a_{ik})u_{k}=0:i=1,..,m. 
\]%
Consider now, for arbitrary $l=1,...,m$, an element $w_{l}=%
\sum_{k=1,..,m}u_{k}x_{kl}\in X_{A}$ which from the first assumption is
nonzero. From the multiplication law of the algebra $X_{A}$ we get%
\[
x_{ij}w_{l}=\sum_{k=1,..,m}u_{k}x_{ij}x_{kl}=%
\sum_{k=1,..,m}u_{k}a_{jk}x_{il}=0\quad \forall i,j=1,...,m. 
\]%
which means that the nonzero elements $w_{l}$ are properly nilpotent and
therefore from the Def.~\ref{df2} the algebra $X_{A}$ is not semisimple.
\end{proof}

From this theorem it follows

\begin{corollary}
\label{col76}
If the algebra $X_{A}$ is semisimple and $\det (A)=0,$ then the vectors $%
\{x_{ij}:i,j=1,...,m\}$ which span the algebra $X_{A}$ are linearly
dependent.
\end{corollary}

In this case the question naturally arises, how to reduce the set of
linearly dependent vectors $\{x_{ij}:i,j=1,...,m\}$ to the set of linearly
independent ones. If the algebra $X_{A}$ is semisimple then we have the
following method of constructing basis of $X_{A}.$

\begin{theorem}
\label{th77}
Let the algebra $X_{A}$ be a semisimple algebra such that 
\[
X_{A}=\Span_{%
\mathbb{C}
}\{x_{ij}:i,j=1,...,m\}, 
\]%
where 
\[
x_{ij}x_{kl}=a_{jk}x_{il}. 
\]%
and the matrix $A=(a_{ij})$ is diagonalizable i.e.%
\[
Z^{-1}AZ=diag(\lambda _{1}\neq 0,...,\lambda _{p}\neq
0,0,..,0)\Leftrightarrow \sum_{jk}z_{ij}^{-1}a_{jk}z_{kl}=\lambda _{i}\delta
_{il},\quad Z\in M(m,%
\mathbb{C}
). 
\]%
So if $p=m,$ then the matrix $A$ is invertible. Define a new elements of the
algebra $X_{A}$ 
\[
y_{sr}=\sum_{jk}z_{rj}^{-1}x_{ij}z_{is}. 
\]%
These elements have the following properties 
\[
y_{sr}=y_{rs}=0\qquad \forall s=1,...,m,\quad r>p, 
\]%
the remaining non-zero vectors $\{y_{ij}:i,j=1,...,p\}$ form the basis of
the algebra $X_{A},$which we call a reduced basis, and they satisfy the
following multiplication rule 
\[
y_{ij}y_{kl}=\lambda _{j}\delta _{jl}y_{il},\quad i,j,k,l=1,...,p. 
\]%
A simple rescaling of the basis vectors $\{y_{ij}:i,j=1,...,p\}$ of the form%
\[
y_{ij}\rightarrow f_{ij}=\frac{1}{\sqrt{\lambda _{i}\lambda _{j}}}y_{ij} 
\]%
gives a new basis of the algebra $X_{A}$, which satisfies the matrix
multiplication rule 
\[
f_{ij}f_{kl}=\delta _{jl}f_{il},\quad i,j,k,l=1,...,p 
\]%
and this proves that the algebra $X_{A}$ is isomorphic with the matrix
algebra $M(p,%
\mathbb{C}
).$
\end{theorem}

\section{}
\label{AppC}

\bigskip In this Appendix we describe the method of reducing the algebra
generated by operators representing a finite group in a given
representation, to the direct sum of matrix algebras.

\begin{notation}
\label{not78}
Let $G$ be a finite group of order $\left\vert G\right\vert =n$ which has $r$
classes of conjugated elements. Then $G$ has exactly $r$ inequivalent,
irreducible representations, in particular $G$ has exactly $r$ inequivalent,
irreducible matrix representations. Let 
\[
\varphi ^{\alpha }:G\rightarrow \Hom(V^{\alpha }),\qquad \alpha
=1,2,....,r,\qquad \dim V^{\alpha }=w^{\alpha } 
\]%
be all inequivalent, irreducible representations of $G$ and let chose these
representations to be all unitary (always possible) i.e. 
\[
\varphi ^{\alpha }(g)=(\varphi _{ij}^{\alpha }(g)),\qquad
i,j=1,2,....,w^{\alpha },\qquad (\varphi _{ij}^{\alpha }(g))^{\dagger}=(\varphi
_{ij}^{\alpha }(g))^{-1} 
\]%
and $V^{\alpha }$ are corresponding representation spaces.
\end{notation}

The matrix elements $\varphi _{ij}^{\alpha }(g)$ will play a crucial role in
the following.

Any complex finite-dimensional representation $D:G\rightarrow \Hom(V)$ of the
finite group $G$, where $V$ is a complex linear space ($\dim V=w)$,
generates a algebra $A_{V}[G] \subset \Hom(V)$ which homomorphic to the
group algebra $%
\mathbb{C}
\lbrack G]$ and in particular isomorphic if the representation $D$ is
faithful. Obviously 
\[
A_{V}[G]=\Span_{%
\mathbb{C}
}\{D(g),\quad g\in G\}. 
\]%
We have also the following decomposition of the representation $%
D:G\rightarrow \Hom(V)$ 
\[
D=\bigoplus _{\alpha =1}^{r}k_{\alpha }\varphi ^{\alpha },\quad V=\bigoplus
_{\alpha =1}^{r}k_{\alpha }V^{\alpha },\quad k_{\alpha }\in 
\mathbb{N}
\cup \{0\}, 
\]%
where $k_{\alpha }$ is the multiplicity of the irreducible representation $%
D^{\alpha }$ in $D.$If the operators $D(g)$ are linearly independent, then
they form a basis of the algebra $A_{V}[G]$ and in this case $\dim
A_{V}[G]=\left\vert G\right\vert $. It is also possible, using matrix
irreducible representations, to construct a new basis which has a remarkable
properties, very useful in applications of representation theory. Below we
define the new basis \ and we use it \ in the study of the linear
independence of the group operators $D(g),$ $g\in G$.

\begin{definition}
\label{def79}
Let $D:G\rightarrow \Hom(V)$ be an unitary representation of a finite group $%
G $ and let $\varphi ^{\alpha }:G\rightarrow \Hom(V^{\alpha })$ be all
inequivalent, irreducible representations of $G$ . Define a matrix operators
in the following way 
\[
E_{ij}^{\alpha }=\frac{w^{\alpha }}{n}\sum_{g\in G}\varphi _{ji}^{\alpha
}(g^{-1})D(g),\quad \alpha =1,2,...,r,\quad i,j=1,2,..,w^{\alpha },\quad
E_{ij}^{\alpha }\in A_{V}[G]\subset \Hom(V).\qquad 
\]
\end{definition}

\begin{remark}
\label{rem80}
\bigskip In this definition we do not assume that the operators $D(g),$ $%
g\in G$ are linearly independent.
\end{remark}

The matrix operators have noticeable properties listed in the

\begin{theorem}
\label{th81}
\begin{enumerate}[I)]
\item There are exactly $\left\vert G\right\vert =n$ \ \ operators $%
E_{ij}^{\alpha }$ but in general they need not to be distinct and we have 
\[
D(g)=\sum_{ij\alpha }\varphi _{ij}^{\alpha }(g)E_{ij}^{\alpha } 
\]%
\ \ \qquad \qquad\ \ \item the operators $E_{ij}^{\alpha }$ are orthogonal
with respect to the Hilbert-Schmidt scalar product in the space $\Hom(V).$ 
\[
(E_{ij}^{\alpha },E_{kl}^{\beta })=\tr((E_{ij}^{\alpha })^{\dagger}E_{kl}^{\beta
})=k_{\alpha }\delta ^{\alpha \beta }\delta _{ik}\delta _{jl},
\]%
where $k_{\alpha }$ \ is the multiplicity of the irreducible representation $%
\varphi ^{\alpha }$ in $D$ and it does not depend on$i,j=1,2,....,w^{\alpha
}.$

\item the operators $E_{ij}^{\alpha }$ satisfy the following composition rule%
\[
E_{ij}^{\alpha }E_{kl}^{\beta }=\delta ^{\alpha \beta }\delta
_{jk}E_{il}^{\alpha }, 
\]%
\ \ \ \ in particular $E_{ii}^{\alpha }$ are orthogonal projections.

\item \label{IV}If the irreducible representation $\varphi ^{\alpha }$ is included in
the representation $D$, then for a fixed $j=1,2,....,w^{\alpha }$ the
operators $E_{ij}^{\alpha }:i=1,2,....,w^{\alpha }$ are the basis of the
representation $\varphi ^{\alpha }$ i.e. we have 
\[
D(h)E_{ij}^{\alpha }=\sum_{k=1}^{w^{\alpha }}\varphi _{kj}^{\alpha
}(h)E_{kj}^{\alpha }. 
\]
\end{enumerate}
\end{theorem}

\begin{remark}
\label{rem82}
Note that from point II) of the theorem it follows that if $k_{\alpha }=0$
for some $\alpha \in \{1,...,r\}$ (that is irreducible representation $%
\varphi ^{\alpha }$ is not included in $D)$, then necessarily $%
E_{ij}^{\alpha }=0$ for all $i,j=1,2,....,w^{\alpha }$.This important
property will be considered in details below. From this part of the theorem
it follows also that the equations.%
\[
E_{ij}^{\alpha }=\frac{w^{\alpha }}{n}\sum_{g\in G}\varphi _{ji}^{\alpha
}(g^{-1})D(g) 
\]%
describe transformation of orthogonalization of operators $D(g),$ $\ g\in G$
in the space \bigskip $\Hom(V)$ with the Hilbert-Schmidt scalar product.
\end{remark}

\bigskip Using elementary methods or the famous Peter-Weyl theorem one can
prove

\begin{theorem}
\label{th84}
The operators $\{E_{ij}^{\alpha }:\quad \alpha =1,2,...,r,\quad
i,j=1,2,..,w^{\alpha }\}$ are linearly independent if and only if the
operators $\{D(g),$ $\ g\in G\}$ are linearly independent. Obviously in this
case both sets form a bases of the algebra $A_{V}[G]$ .
\end{theorem}

\bigskip From Th.81 it follows, in particular, that if the
irreducible representation $\varphi ^{\alpha }$ of $G$
has multiplicity one in the representation $D$ then the operators 
$E_{jj}^{\alpha },$ $j=1,2,..,w^{\alpha }$ are a
hermitian projectors of rank one. Such a projectors have the following
properties

\begin{theorem}
\textbf{Let }$P=(p_{ij})\in M(n,%
\mathbb{C}
)$ and $P^{2}=P=P^{+}$, $\rank(P)=1$ and let $%
P_{q}=(p_{iq}):i=1,...n$ be a non-zero coloumn of P, then 
\[
0<p_{qq}\leq 1. 
\]
\end{theorem}

\begin{corollary}
\bigskip If the irreducible representation $\varphi ^{\alpha }$%
 of $G$ has multiplicity one in the representation $D$, then for any $j=1,2,..,w^{\alpha },$ there exist a
non-zero column $q$ in the projector $E_{jj}^{\alpha }$,
such that $(E_{jj}^{\alpha })_{qq}>0.$
\end{corollary}

\bigskip

The structure of the algebra $A_{V}[G],$in
particular its dimension, is the following .

\begin{theorem}
\[
A_{V}[G]=\bigoplus _{\alpha }span_{%
\mathbb{C}
}\{E_{ij}^{\alpha }\}:\varphi ^{\alpha }\in D,
\]%
in particular 
\[
\dim A_{V}[G]=\sum_{\alpha :\varphi ^{\alpha }\in D}(\dim \varphi ^{\alpha
})^{2}.
\]%
The operators $D(g):g\in G$ are linearly independent, in
this case $\dim A_{V}[G]=|G|,$ iff each irreducible
representations $\varphi ^{\alpha },\alpha =1,2,....,r$, appears in
the decomposition $D=\oplus _{\alpha =1}^{r}k_{\alpha }\varphi ^{\alpha }$%
 e.i. $k_{\alpha }\geq 1,\quad \forall \alpha =1,2,...,r.$
\end{theorem}

\bigskip From this theorem it follows that in order to check the linear
independence of the operators $D(g),$ $\ g\in G$ in the representation
algebra algebra $A_{V}[G]$ it is enough to know the multiplicities $%
k_{\alpha }$ of the irreducible representations $D^{\alpha },\alpha
=1,2,....,r$ in the representation $D.$


\newpage

\begin{thebibliography}{27}
\expandafter\ifx\csname natexlab\endcsname\relax\def\natexlab#1{#1}\fi
\expandafter\ifx\csname bibnamefont\endcsname\relax
  \def\bibnamefont#1{#1}\fi
\expandafter\ifx\csname bibfnamefont\endcsname\relax
  \def\bibfnamefont#1{#1}\fi
\expandafter\ifx\csname citenamefont\endcsname\relax
  \def\citenamefont#1{#1}\fi
\expandafter\ifx\csname url\endcsname\relax
  \def\url#1{\texttt{#1}}\fi
\expandafter\ifx\csname urlprefix\endcsname\relax\def\urlprefix{URL }\fi
\providecommand{\bibinfo}[2]{#2}
\providecommand{\eprint}[2][]{\url{#2}}

\bibitem[{\citenamefont{Fulton and Harris}(1991)}]{Fulton}
\bibinfo{author}{\bibfnamefont{W.}~\bibnamefont{Fulton}} \bibnamefont{and}
  \bibinfo{author}{\bibfnamefont{J.}~\bibnamefont{Harris}},
  \emph{\bibinfo{title}{Representation Theory - A First Course}}
  (\bibinfo{publisher}{Springer-Verlag, New York}, \bibinfo{year}{1991}).
  
  \bibitem[{\citenamefont{Landsman}(1998)\citenamefont{Landsman}}]{Landsman}
\bibinfo{author}{\bibfnamefont{N.~P.} \bibnamefont{Landsman}},
  \urlprefix\url{http://arxiv.org/pdf/math-ph/9807030v1.pdf}.
  
  \bibitem[{\citenamefont{Curtis and Reiner}(1991)}]{Curtis}
\bibinfo{author}{\bibfnamefont{C.~W.}~\bibnamefont{Curtis}} \bibnamefont{and}
  \bibinfo{author}{\bibfnamefont{I.}~\bibnamefont{Reiner}},
  \emph{\bibinfo{title}{Representation theory of finite groups and
 associative algebras}}
  (\bibinfo{publisher}{Pure and Applied Mathematics Vol. XI, Interscience Publishers}, \bibinfo{year}{1991}).
  
\bibitem[{\citenamefont{Littlewood}(1991)}]{Littlewood}
\bibinfo{author}{\bibfnamefont{D.~E.}~\bibnamefont{Littlewood}},
  \emph{\bibinfo{title}{The Theory of Group Characters and Matrix Representations of Groups}}
  (\bibinfo{publisher}{Oxford at the Clarendon Press}, \bibinfo{year}{1991}).
  
  \bibitem[{\citenamefont{Boerner}(1970)}]{Boerner}
\bibinfo{author}{\bibfnamefont{H.}~\bibnamefont{Boerner}},
  \emph{\bibinfo{title}{Representations of Groups}}
  (\bibinfo{publisher}{North-Holland}, \bibinfo{year}{1970}).
  
  \bibitem[{\citenamefont{Studzinski et~al.}(2013)\citenamefont{Studzinski,Horodecki and Mozrzymas}}]{Studzinski1}
\bibinfo{author}{\bibfnamefont{M.} \bibnamefont{Studzi{\'n}ski}},
  \bibinfo{author}{\bibfnamefont{M.}~\bibnamefont{Horodecki}}, \bibnamefont{and}
  \bibinfo{author}{\bibfnamefont{M.}~\bibnamefont{Mozrzymas}},
  \bibinfo{journal}{J. Phys. A: Math. Theor.} \textbf{\bibinfo{volume}{46}},
  \bibinfo{pages}{395303} (\bibinfo{year}{2013}).
  
  \bibitem[{\citenamefont{Studzinski et~al.}(2014)\citenamefont{Studzinski, Cwiklinski, Horodecki and Mozrzymas}}]{Studzinski2}
\bibinfo{author}{\bibfnamefont{M.} \bibnamefont{Studzi{\'n}ski}},
\bibinfo{author}{\bibfnamefont{P.}~\bibnamefont{{\'C}wikli{\'n}ski}},
  \bibinfo{author}{\bibfnamefont{M.}~\bibnamefont{Horodecki}}, \bibnamefont{and}
  \bibinfo{author}{\bibfnamefont{M.}~\bibnamefont{Mozrzymas}},
  \urlprefix\url{http://arxiv.org/abs/1311.2772}.
  
  \bibitem[{\citenamefont{Horodecki et~al.}(1996)\citenamefont{Horodecki, Horodecki and Horodecki}}]{Horodecki2}
\bibinfo{author}{\bibfnamefont{M.} \bibnamefont{Horodecki}},
  \bibinfo{author}{\bibfnamefont{P.}~\bibnamefont{Horodecki}}, \bibnamefont{and}
  \bibinfo{author}{\bibfnamefont{R.}~\bibnamefont{Horodecki}},
  \bibinfo{journal}{Phys. Lett. A} \textbf{\bibinfo{volume}{223}},
  \bibinfo{pages}{1-8} (\bibinfo{year}{1996}).
  
   \bibitem[{\citenamefont{Eggeling et~al.}(1993)\citenamefont{Eggeling and Werner}}]{Eggeling1}
\bibinfo{author}{\bibfnamefont{T.} \bibnamefont{Eggeling}},
  \bibinfo{author}{\bibfnamefont{R.~F.}~\bibnamefont{Werner}},
  \bibinfo{journal}{Phys. Rev. A} \textbf{\bibinfo{volume}{63}},
  \bibinfo{pages}{042111} (\bibinfo{year}{2000}).
  
\bibitem[{\citenamefont{Peres}(1996)\citenamefont{Peres}}]{Peres1}
\bibinfo{author}{\bibfnamefont{A.} \bibnamefont{Peres}},
  \bibinfo{journal}{Phys. Rev. Lett.} \textbf{\bibinfo{volume}{77}},
  \bibinfo{pages}{1413-1415} (\bibinfo{year}{1996}).
  
  \bibitem[{\citenamefont{Zhang et~al.}(1993)\citenamefont{Zhang, Kauffman and Werner}}]{Zhang1}
\bibinfo{author}{\bibfnamefont{Y.} \bibnamefont{Zhang}},
  \bibinfo{author}{\bibfnamefont{L.~H.}~\bibnamefont{Kauffman}}, \bibnamefont{and}
  \bibinfo{author}{\bibfnamefont{R.~F.}~\bibnamefont{Werner}},
  \bibinfo{journal}{Int.J.Quant.Inf.} \textbf{\bibinfo{volume}{5}},
  \bibinfo{pages}{469-507} (\bibinfo{year}{2007}).
  
   \bibitem[{\citenamefont{Tura et~al.}(2012)\citenamefont{Tura}}]{Tura1}
\bibinfo{author}{\bibfnamefont{J.} \bibnamefont{Tura}},
  \bibinfo{author}{\bibfnamefont{R.}~\bibnamefont{Augusiak}},
  \bibinfo{author}{\bibfnamefont{P.}~\bibnamefont{Hyllus}}, 
  \bibinfo{author}{\bibfnamefont{J.}~\bibnamefont{Samsonowicz}},
  \bibnamefont{and}
  \bibinfo{author}{\bibfnamefont{M.}~\bibnamefont{Lewenstein}},
  \bibinfo{journal}{Phys. Rev. A} \textbf{\bibinfo{volume}{85}},
  \bibinfo{pages}{ 060302(R)} (\bibinfo{year}{2012}).
  
  \bibitem[{\citenamefont{Augusiak et~al.}(2012)\citenamefont{Augusiak}}]{Augusiak1}
\bibinfo{author}{\bibfnamefont{R.} \bibnamefont{Augusiak}},
  \bibinfo{author}{\bibfnamefont{J.}~\bibnamefont{Tura}},
  \bibinfo{author}{\bibfnamefont{J.}~\bibnamefont{Samsonowicz}}, \bibnamefont{and}
  \bibinfo{author}{\bibfnamefont{M.}~\bibnamefont{Lewenstein}},
   \bibinfo{journal}{Phys. Rev. A} \textbf{\bibinfo{volume}{86}},
  \bibinfo{pages}{042316} (\bibinfo{year}{2012}).
  
  \bibitem[{\citenamefont{Montanaro}()}]{Montanaro1}
\bibinfo{author}{\bibfnamefont{A.} \bibnamefont{Montanaro}},
  \emph{\bibinfo{title}{Weak multiplicativity for random quantum channels}},
  \urlprefix\url{http://arxiv.org/abs/1112.5271}.
  
  \bibitem[{\citenamefont{Brandao}()}]{Brandao1}
\bibinfo{author}{\bibfnamefont{F.~G.~S.~L.} \bibnamefont{Brandao}},
\bibinfo{author}{\bibfnamefont{A.~W.} \bibnamefont{Harrow}},
  \emph{\bibinfo{title}{Quantum de Finetti Theorems under Local Measurements with Applications}},
  \urlprefix\url{http://arxiv.org/abs/1210.6367}.
  
  \bibitem[{\citenamefont{Gavarini et~al.}(2012)\citenamefont{Gavarini and Papi}}]{Gavarini1}
\bibinfo{author}{\bibfnamefont{F.} \bibnamefont{Gavarini}},
  \bibinfo{author}{\bibfnamefont{P.}~\bibnamefont{Papi}},
  \bibinfo{journal}{Journal of Algebra} \textbf{\bibinfo{volume}{194}},
  \bibinfo{pages}{275-298} (\bibinfo{year}{1997}).
  
  \bibitem[{\citenamefont{Brauer}(2012)\citenamefont{Gavarini}}]{Brauer1}
\bibinfo{author}{\bibfnamefont{R.} \bibnamefont{Brauer}},
  \bibinfo{journal}{Annals of Mathematics} \textbf{\bibinfo{volume}{38}} \textbf{\bibinfo{No.}{4}},
  \bibinfo{pages}{857-872} (\bibinfo{year}{1937}).
  
  \bibitem[{\citenamefont{Pan}(1995)\citenamefont{Pan}}]{Pan1}
\bibinfo{author}{\bibfnamefont{F.} \bibnamefont{Pan}},
  \bibinfo{journal}{J. Phys.  A: Math. Gen.} \textbf{\bibinfo{volume}{28}},
  \bibinfo{pages}{3139-3156} (\bibinfo{year}{1995}).
  
  \bibitem[{\citenamefont{Koike}(1989)\citenamefont{Koike}}]{Koike1}
\bibinfo{author}{\bibfnamefont{K.} \bibnamefont{Koike}},
  \bibinfo{journal}{Advances Math.} \textbf{\bibinfo{volume}{74}},
  \bibinfo{pages}{57-86} (\bibinfo{year}{1989}).
  
  \bibitem[{\citenamefont{Turaev}(1989)\citenamefont{Turaev}}]{Turaev1}
\bibinfo{author}{\bibfnamefont{V.} \bibnamefont{Turaev}},
  \bibinfo{journal}{Izv. Akad. Nauk SSSR} \textbf{\bibinfo{volume}{53}},
  \bibinfo{pages}{1073-1107} (\bibinfo{year}{1989}).
  
  \bibitem[{\citenamefont{Brundan}(2009)\citenamefont{Brundan}}]{Brundan1}
\bibinfo{author}{\bibfnamefont{J.} \bibnamefont{Brundan}},
\bibinfo{author}{\bibfnamefont{C.} \bibnamefont{Stroppel}},
  \bibinfo{journal}{Advances Math.} \textbf{\bibinfo{volume}{231}} ,
  \bibinfo{pages}{709-773} (\bibinfo{year}{2012}).
  
  \bibitem[{\citenamefont{Benkart et al.}(1994)\citenamefont{Benkart}}]{Benkart1}
\bibinfo{author}{\bibfnamefont{G.} \bibnamefont{Benkart}},
\bibinfo{author}{\bibfnamefont{M.} \bibnamefont{Chakrabarti}},
\bibinfo{author}{\bibfnamefont{T.} \bibnamefont{Halverson}},
\bibinfo{author}{\bibfnamefont{R.} \bibnamefont{Leduc}},
\bibinfo{author}{\bibfnamefont{C.} \bibnamefont{Lee}}, \bibnamefont{and}
\bibinfo{author}{\bibfnamefont{J.} \bibnamefont{Stroomer}}
  \bibinfo{journal}{J. Algebra} \textbf{\bibinfo{volume}{166}},
  \bibinfo{pages}{529-567} (\bibinfo{year}{1994}).
  
  
   \bibitem[{\citenamefont{Cox et~al.}(1993)\citenamefont{Cox, Visscher, Doty and Martin}}]{Cox1}
\bibinfo{author}{\bibfnamefont{A.} \bibnamefont{Cox}},
  \bibinfo{author}{\bibfnamefont{M.}~\bibnamefont{Visscher}},
  \bibinfo{author}{\bibfnamefont{S.}~\bibnamefont{Doty}}, \bibnamefont{and}
  \bibinfo{author}{\bibfnamefont{P.}~\bibnamefont{Martin}},
  \bibinfo{journal}{ Journal of Algebra} \textbf{\bibinfo{volume}{320}},
  \bibinfo{pages}{169-212} (\bibinfo{year}{2008}).

\end{thebibliography}
\end{document}